\documentclass[12pt]{iopart}

\usepackage{amsfonts,amsbsy,amsthm,color}
\usepackage{graphicx}
\usepackage{setstack}
\usepackage[mathscr]{euscript}

\newcommand{\iP}{\widetilde{\mathcal{P}}}
\newcommand{\hP}{\widehat{\mathcal{P}}}
\newcommand{\cM}{\mathcal{M}}
\newcommand{\cD}{\mathcal{D}}
\newcommand{\cB}{\mathcal{B}}

\newcommand{\bg}{\bar{\gamma}}
\newcommand{\link}{\mathrm{link}}
\newcommand{\enc}{\mathrm{enc}}
\newcommand{\tup}{\mathrm{tup}}
\newcommand{\wt}{\mathrm{wt}}
\newcommand{\sgn}{\mathrm{sgn}}
\newcommand{\PO}{\mathrm{PO}}
\newcommand{\PPO}{\mathrm{PPO}}
\newcommand{\Lyndon}{\mathrm{L}_2(n)}
\newcommand{\diag}{\mathrm{diag}}

\newcommand{\A}{\mathbf{A}}
\newcommand{\U}{\mathbf{U}}
\newcommand{\ML}{\mathbf{L}}
\newcommand{\UI}{\mathrm{\bf I}}  
\newcommand{\VS}{\boldsymbol{\sigma}}
\newcommand{\BS}{\boldsymbol{\Sigma}}
\newcommand{\T}{\mathbf{T}}
\newcommand{\B}{\mathbf{B}}
\renewcommand{\S}{\mathbf{S}}
\renewcommand{\H}{\mathbf{H}}
\newcommand{\F}{\mathbf{F}}
\newcommand{\D}{\mathbf{D}}
\newcommand{\1}{\mathbf{1}}

\newtheorem{theorem}{Theorem}
\newtheorem{corollary}{Corollary}
\newtheorem{proposition}{Proposition}
\newtheorem{lemma}{Lemma}

\begin{document}

\title[The characteristic polynomial of binary quantum graphs]{Complete dynamical evaluation of the characteristic polynomial of binary quantum graphs}

\author{J.M. Harrison$^1$ and T. Hudgins$^2$}

\address{$^1$ Baylor University, Department of Mathematics,
	Sid Richardson Building, 1410 S. 4th Street, Waco, TX 76706, USA}
\address{$^2$ University of Kansas, Department of Mathematics,
	405 Snow Hall, 1460 Jayhawk Blvd., Lawrence, KS 66045, USA}
\ead{jon\_harrison@baylor.edu, thudgins@ku.edu}
\vspace{10pt}
\begin{indented}
\item[]May 2022
\end{indented}

\begin{abstract}
We evaluate the variance of coefficients of the characteristic polynomial for binary quantum graphs using a dynamical approach.  This is the first example where a spectral statistic can be evaluated in terms of periodic orbits for a system with chaotic classical dynamics without taking the semiclassical limit, which here is the limit of large graphs.  The variance depends on the sizes of particular sets of primitive pseudo orbits (sets of distinct primitive periodic orbits): the set of primitive pseudo orbits without self-intersections and the sets of primitive pseudo orbits with a fixed number of self-intersections, all of which consist of two arcs of the pseudo orbit crossing at a single vertex. To show other pseudo orbits do not contribute we give two arguments. The first is based on a reduction of the variance formula from a sum over pairs of primitive pseudo orbits to a sum over pseudo orbits where no bonds are repeated.  The second employs a parity argument for the Lyndon decomposition of words.  
For families of binary graphs, in the semiclassical limit, we show the pseudo orbit formula approaches a universal constant independent of the coefficient of the polynomial.  This is obtained by counting the total number of primitive pseudo orbits of a given length.  
\end{abstract}

%
\vspace{2pc}
\noindent{\it Keywords}: quantum graphs, quantum chaos, semiclassical methods, Lyndon words


\section{Introduction}
\label{sec:intro}
The dynamical approach to spectral problems started fifty years ago with Gutzwiller's trace formula \cite{G71,G90}.  Gutzwiller derived a semiclassical expansion for the density of states as a sum over classical periodic orbits for a wide class of quantum systems, including those where the classical dynamics are chaotic.  The trace formula has formed the basis to study spectral properties of quantum chaotic systems.  For instance, it provides a mechanism to analyze small parameter asymptotics of the form factor, the Fourier transformation of the two point correlation function.  Zeroth order contributions in this expansion were evaluated by Berry \cite{B85} using the Hannay and Ozorio de Almeida sum rule \cite{HA84}.  It took a further fifteen years before Sieber and Sieber and Richter \cite{S02,SR01} obtained first order terms in the expansion for quantum billiards by considering figure eight orbits with a single self-intersection.  Subsequently, higher order terms were evaluated in quantum graphs by including orbits with multiple self-intersections by Berkolaiko, Schanz and Whitney \cite{BSW02,BSW03} and the scheme was expanded to obtain all orders in the form factor expansion by M\"uller, Heusler, Braun, Haake, and Altland \cite{Metal04}.  
Note that, while we mentioned some results for the form factor, this kind of dynamical analysis, in the semiclassical limit, has been applied in many contexts, see e.g. \cite{BBK01,BHN08,BK12,BH03,DOW05,KM00,KMM01}.  

In this dynamical approach to spectral statistics, a significant role has been played by quantum graph models.  
A quantum graph is a quasi-one-dimensional network of bonds connected at vertices.
Quantum graphs were introduced as a model of a quantum system with chaotic classical dynamics by Kottos and Smilansky \cite{KS99}. They derive a trace formula for quantum graphs that plays the role of the Gutzwiller trace formula, expressing the density of states as a sum over periodic orbits.  However, for quantum graphs, this trace formula is exact. It holds for graphs of finite size rather than only in the semiclassical limit of increasing spectral density, which for quantum graphs is the limit of a sequence of graphs with increasing numbers of bonds.  The first quantum graph trace formula was obtained by Roth \cite{R83}, and general quantum graph trace formulas appear in \cite{BE09,KPS07}.  While the trace formula for quantum graphs is exact, to evaluate contributions to the form factor expansion still requires the semiclassical limit.  Quantum graphs are currently employed in many areas of mathematical physics from Anderson localization and carbon nanotubes, to mesoscopic physics and waveguides; see \cite{BK13} for an introduction to quantum graphs and \cite{GS06} for a review of quantum graphs in quantum chaos. 

In \cite{KS99} Kottos and Smilansky also investigated the coefficients of the characteristic polynomial of a quantum graph.  The characteristic polynomial of a quantum graph with $B$ bonds is,
\begin{equation} 
\label{charpoly}
\det(\U (k) - \zeta \UI) = \sum\limits_{n=0}^{B} a_n \zeta^{B-n} \ .
\end{equation}
The matrix $\U (k)$ is the quantum evolution operator of the graph, a product of the bond scattering matrix and a matrix of phases $\rme^{\rmi k\ML }$ acquired when traversing the bonds where $\ML$ is a diagonal matrix of bond lengths.\footnote{The bond scattering matrix and quantum evolution operator are often written as $2B\times 2B$ matrices for a graph with $B$ bonds which is consistent with the notation here when two directed bonds run between every pair of connected vertices, one in each direction.}  The equation $ \det(\U (k) - \UI)=0$ is a secular equation for the quantum graph; solutions are the $k$-spectrum of the graph, square roots of the eigenvalues \cite{KS99}.
Significantly, the coefficients $a_n$ can also be expressed in terms of periodic orbits on the graph, see section \ref{sec:orbits}. The formula for the coefficients $a_n$, like the graph trace formula, does not require a semiclassical limit.  The $n$-th coefficient of the characteristic polynomial is written as a sum over collections of periodic orbits where the total number of bonds in the set of orbits is $n$.  These collections of periodic orbits are referred to as pseudo orbits, so coefficients are expressed as a sum over a finite collection of pseudo orbits.  Kottos and Smilansky describe this in \cite{KS99} without an explicit formula, and a minimal pseudo orbit expansion was obtained in \cite{BHJ12}.  Semiclassical pseudo orbit expansions derived from the Gutzwiller trace formula have been applied to evaluate correlators of the level density in chaotic quantum systems \cite{Metal05,NM09}.

Averaging $a_n$ over the spectral parameter $k$, one obtains zero (except for $a_0$ where the average is one).  In this article we show that the first nontrivial moment of the coefficients, the variance, can be fully evaluated dynamically for families of binary graphs without taking a semiclassical limit.  Binary graphs were introduced by Tanner \cite{T00,T01,T02}, see section \ref{sec:binary}. They are directed graphs with $V=p\cdot 2^r$ vertices and $B= p\cdot 2^{r+1}$ bonds with two incoming and two outgoing bonds at each vertex; $p$ is odd.
When $p=1$ there is a natural labeling of the vertices by binary words of length $r$ and of edges by binary words of length $r+1$.  The first $r$ letters of an edge label is the initial vertex of the bond and the last $r$ letters are the terminal vertex. 
In \cite{T02} the asymptotics of the variance of the coefficients of the characteristic polynomial are investigated in the semiclassical limit for families of binary graphs with $p$ fixed.  They are seen to approach a family-dependent constant independent of the coefficient of the characteristic polynomial.  This constant is obtained from a diagonal approximation using a matrix permanent.   In \cite{BHS19} the variance of families of $q$-nary graphs with $p=1$ was investigated semiclassically via pseudo orbits by evaluating a diagonal contribution to the variance.  

We now describe the main results.  For binary graphs, the variance of the coefficients of the characteristic polynomial is determined by the sizes of certain sets of primitive pseudo orbits.  A primitive periodic orbit is an orbit that is not a repetition of a shorter orbit, and a primitive pseudo orbit is a set of distinct primitive periodic orbits.  A pseudo orbit has an encounter, or self-intersection, if a section of the orbit is repeated in the orbit or in another periodic orbit in the set.  The repeated section could include one or more bonds in the graph, so the orbit intersects itself and follows the original path for some number of bonds, or the repeated section can be a single vertex where the periodic orbit or orbits intersect, in which case we say the encounter has length zero.  An $\ell$-encounter is a section that is repeated $\ell$ times in one or more orbits, see section \ref{sec:enc} for more details.  Orbits and pseudo orbits can have multiple encounters with possibly different values of $\ell$ and different lengths for each encounter.  However, it turns out that the variance only depends on the sizes of sets of primitive pseudo orbits with particular types of encounters.   

\begin{proposition}\label{thm:main}  Consider a binary graph with $V=p\cdot 2^r$ vertices where $p$ is odd.  The variance of the $n$-th coefficient of the characteristic polynomial is, 
	\begin{equation}
	\label{variancefinal}
	= \frac{1}{2^{n}} \left( |\mathcal{P}_0^n| + \sum_{N=1}^n 2^N \, |\hP_{N}^n| \right) \ ,
	\end{equation}
	where $\mathcal{P}_0^n$ is the set of primitive pseudo orbits with $n$ bonds and no self-intersections and $\hP_{N}^n$ is the set of primitive pseudo orbits with $n$ bonds and $N$ self-intersections, all of which are $2$-encounters of length zero.
\end{proposition}
Comparing the variance evaluated via pseudo orbits using proposition \ref{thm:main} to numerical results we find agreement to four decimal places, see section \ref{sec:examples}.

While this result shares some features of the semiclassical expansion of the form factor it is surprisingly different.
  As in semiclassical expansions of the form factor, an important role is played by orbits with self-intersections.  However, in the form factor expansion all orbits contribute to zeroth order and higher order contributions are obtained from all orbits with progressively more self-intersections.   In contrast, in proposition \ref{thm:main}, only pseudo orbits where all the self-intersections are of length zero contribute.

In \cite{BHS19} Band, Harrison and Sepanski evaluate a diagonal contribution to the variance dynamically for the family of binary graphs with $p=1$ by counting the total number of primitive pseudo orbits of length $n$ using Lyndon decompositions of binary words of length $n$.  In the semiclassical limit of large graphs the variance converges to this diagonal contribution, in agreement with \cite{T02}.  It is therefore important to understand how this is possible, given that many pseudo orbits do not appear in proposition \ref{thm:main}.    
 In the semiclassical limit of large graphs, for a family of graphs defined by the choice of $p$, the sum in parenthesis in (\ref{variancefinal}) converges to the total number of primitive pseudo orbits on the graph, see section \ref{sec:semiclassical}.  
 So while the variance is determined only by the numbers of primitive pseudo orbits with self-intersections of zero length weighted by certain factors, it turns out that the sizes of these sets combined with the weights produces, in the semiclassical limit, the total number of pseudo orbits.

 To evaluate this semiclassical variance for $p\neq 1$ we count primitive orbits and pseudo orbits on binary graphs.
\begin{proposition}
	\label{thmNumPPOs}
	Let $p>0$ be an odd number and consider the binary graph with $V=p \cdot 2^r$ vertices.  Then the number of primitive pseudo orbits of length $n > p$ is 
	\begin{equation}
	\PPO_{p}(n) = C_{p} \cdot 2^{n-1} \ ,
	\label{eq:PPOpq}
	\end{equation}
	where $C_p$ is a constant, $1 \leq C_{p} \leq \frac{3}{2} (p-1) $ for $p>1$, and $C_1=1$.
\end{proposition}
The family-dependent constant $C_p$ is evaluated from the cycle decomposition of a generalized $p\times p$ permutation matrix in section \ref{sec:counting prim PPOs}.  
This leads to the following corollary.
\begin{corollary}\label{cor:asymp}
	Let $p>0$ be an odd number and consider the family of binary graphs with $V=p\cdot 2^r$ vertices and $B=p\cdot 2^{r+1}$ bonds. Fixing the ratio $n/B$,
	\begin{equation}
	\lim_{r\to\infty}
	\langle |a_n|^2 \rangle_k = \frac{C_{p}}{2} \ .
	\label{cor:varriance assymp}
	\end{equation}
	Consequently, for the family of binary graphs with $p=1$, we have
	$\lim_{r\to\infty} \langle |a_n|^2 \rangle_k = 1/2$.
\end{corollary}
While it was observed in \cite{T02} and \cite{BHS19} that the variance for families of binary quantum graphs approaches a constant value, the argument to obtain the constant, in both cases, makes some form of diagonal assumption, although the arguments are different in each case.  The corollary establishes the result without making such an approximation.


The results presented are a significant advance, evaluating a spectral statistic for a classically chaotic system without using the semiclassical limit.  However, it is important also to point out why it was possible to obtain this for binary quantum graphs while such results are elusive in other systems.  Firstly, quantum graphs have a very strong form of symmetry that allows the coefficients of the characteristic polynomial and also the trace formula, to be expressed as a sum of dynamical quantities without taking the semiclassical limit.  This is why they were proposed as a model for quantum chaos by Kottos and Smilansky \cite{KS99}.  For quantum graphs the coefficients of the characteristic polynomial also exhibit an exact Riemann-Siegel lookalike formula.  The Riemann-Siegel lookalike formula originally discovered by Berry and Keating \cite{BK90,BK92,K92} and in another form by Bogomolny \cite{B92} is a remarkable resummation procedure that connects long and short orbits.  
Having chosen a setting and statistic for which there is an unusually precise dynamical formula, we are then able to exploit a choice of graph, binary graphs, where a lot of dynamical information can be evaluated.   
Tantalizingly, the precise dynamical expression for the variance in proposition \ref{thm:main} suggests the intuition of Sieber and Sieber and Richter \cite{S02,SR01}, when they started by finding partner orbits with $2$-encounters, is also key to unlocking the quantum spectral statistics of classically chaotic systems without the semiclassical limit.  

The article is organized as follows.  In section \ref{sec:background} we introduce binary quantum graphs and the expansion of coefficients of the characteristic polynomial as a sum over primitive pseudo orbits.  Section \ref{sec:variance} proves proposition \ref{thm:main} by evaluating the contribution to the variance from an irreducible pseudo orbit and its partner orbits of the same length.  
In section \ref{sec:contributions} we present a cancellation scheme for the primitive but not irreducible pseudo orbits using a parity argument.  
In section \ref{sec:semiclassical} we take the semiclassical limit of proposition \ref{thm:main} for families of binary graphs which relates the variance of the $n$-th coefficient of the characteristic polynomial to the total number of primitive pseudo orbits with $n$ bonds.
In section \ref{sec:counting} we prove proposition \ref{thmNumPPOs} for the number of primitive pseudo orbits of $n$ bonds which, when combined with the semiclassical formula from section \ref{sec:semiclassical}, establishes corollary \ref{cor:asymp}. We illustrate the results with a number of examples in section \ref{sec:examples} which we compare with numerical calculations of the variance.  
We summarize the results in section \ref{sec:conclusions}.  The appendices describe how the sizes of the sets of pseudo orbits used to compare numerical calculations of the variance with proposition \ref{thm:main} were evaluated.

\section{Background}
\label{sec:background}

\subsection{Binary Graphs}
\label{sec:binary}

A \textit{directed graph} (or graph), $\Gamma$, is a set of \textit{vertices}, $\mathscr{V}= \{0,1,\dots,V-1\}$, connected by a set of \textit{bonds} $\mathscr{B}$.  A bond $b\in \mathscr{B}$ is an ordered pair of vertices $ b=(i,j)$ with $i,j \in\mathscr{V}$.
We denote the number of vertices by $V=|\mathscr{V}|$ and the number of bonds by $B=|\mathscr{B}|$.
The \textit{origin} and \textit{terminus} of a bond can be specified via functions $o,t:\mathscr{B} \to \mathscr{V}$ such that $b=(o(b), t(b))$.
Two vertices, $i,j\in\mathscr{V}$, are \textit{adjacent} if at least one of the ordered pairs $(i,j)$ or $(j,i)$ is in $\mathscr{B}$; we write $i\sim j$.  A bond $b$ is \textit{outgoing} at $v$ if $o(b)=v$ and $b$ is \textit{incoming} at $v$ if $t(b)=v$.
The number of incoming and outgoing bonds at $v$ are denoted $d_v^{\textrm{in}}$ and $d_v^{\textrm{out}}$, respectively, and $d_v = d_v^{\textrm{in}} + d_v^{\textrm{out}}$ is the \textit{degree} of vertex $v$.  
For our purposes we will always consider a graph to be such a directed graph and we use bond rather than edge as a reminder that the graph is directed.  Figure \ref{figBGp1} shows a (directed) graph with four vertices and eight bonds.

\begin{figure}[htb]
	\centering
	\includegraphics[scale=1, trim=125 510 250 125, clip]{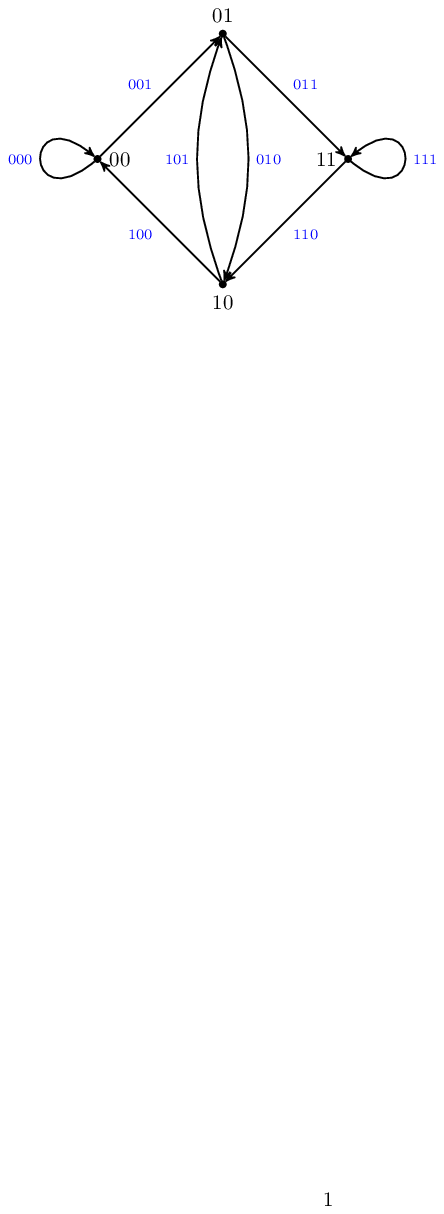}
	\caption{The binary graph with $V=2^2$ vertices and $B=2^3$ bonds} 
	\label{figBGp1}
\end{figure}

A \textit{binary graph}, see \cite{BHS19,T00,T01,T02}, has $V=p\cdot 2^r$ vertices and $B=p\cdot 2^{r+1}$ bonds, where $p$ is odd and $r\in\mathbb{N}$.
The $V\times V$ adjacency matrix, which indicates the presence of a bond connecting vertex $i$ to vertex $j$, has the form,
\begin{equation}\label{eq:defn A_V}
[\A_V]_{i,j} = 
\cases{\delta_{2i,j}+\delta_{2i+1,j}  & for $0\leq i <V/2$ \\ \delta_{2i-V,j}+\delta_{2i+1-V,j} & for $V/2 \leq i < V$\\}
\end{equation}
where $\delta_{i,j}$ is the Kronecker delta and $0\leq j <V$.
So $d_v^{\textrm{in}} = d_v^{\textrm{out}} = 2$, and the graphs are $4$-regular.
The binary graph with $V=2^2$ vertices, see figure \ref{figBGp1}, has the adjacency matrix,
\begin{equation}
\A_4 = \left( \begin{array}{ c c c c }
1 & 1 & 0 & 0 \\
0 & 0 & 1 & 1 \\
1 & 1 & 0 & 0 \\
0 & 0 & 1 & 1 
\end{array} \right) \ .
\end{equation}

When $p=1$ and thus $V=2^r$, a vertex $w\in \{0,1,\dots,V-1\}$ can be written as a word of length $r$ in the binary alphabet $\{0,1\}$,
\begin{equation}
w = w_0 w_1 \cdots w_{r-1}\ . 
\end{equation} 
A vertex labeled by $w_0 w_1 \dots w_{r-1}$ is connected by an outgoing bond to $w_1 w_2 \dots w_{r-1} 0$ and $w_1 w_2 \dots w_{r-1} 1$.
To see this let,
\begin{eqnarray}
\label{ibinary}
i &= w_0 2^{r-1} + w_1 2^{r-2} + \cdots w_{r-2} 2 + w_{r-1}\ ,\\
\label{jbinary}
j & = w_1 2^{r-1} + w_2 2^{r-2} + \cdots w_{r-1} 2 + w_r \ .
\end{eqnarray}
Then $2i + w_r \pmod{V} = j$, so $[\A]_{i,j} =1$. 
So we can label vertices with words of length $r$ and bonds with words of length $r+1$, where the bond labeled by $w_0 w_1 \dots w_{r-1} w_r$ connects the vertex $w_0 w_1 \dots w_{r-1}$ to the vertex $w_1 \dots w_{r-1} w_r$. 
Figure \ref{figBGp1} shows the vertex and bond labels of the binary graph with $V=2^2$.

\subsection{Binary Quantum Graphs}
\label{sec:binaryQG}

A graph becomes a \textit{metric graph} when we assign to each bond $b$ a positive \textit{length} $L_b$.
There are two methods that are commonly used to define a \textit{quantum graph} given a metric graph; details of both can be found in \cite{BK13,GS06}.
Firstly one can define a self-adjoint differential operator acting on functions defined on a set of intervals corresponding to the bond lengths.
The alternative approach, which we adopt here, is to associate to each vertex a unitary \emph{vertex scattering matrix}, $\VS^{(v)}$, that assigns scattering amplitudes from incoming to outgoing bonds at each vertex on the graph.
For a $q$-regular graph, a popular choice is the $q \times q$ Discrete Fourier Transform (DFT) matrix \cite{BHS19,BK13,T00,T01,T02}.  Such a choice is democratic, in the sense that the probability of scattering from incoming bond $b$ to outgoing bond $b'$, which is $|\VS^{(v)}_{b',b}|^2$, is equal for all pairs of bonds $b,b'$ meeting at $v$. For some dimensions, it is also possible to generate unitary scattering matrices with the democratic property from vertex conditions of a self-adjoint Laplace operator on a metric graph \cite{TC11}. 
A vertex scattering matrix which corresponds to Neumann-like (or standard) vertex conditions of a Laplace operator on a metric graph \cite{KS99}, is
\begin{equation}
\VS_{b',b}^{(v)} = \frac{2}{d_v} - \delta_{b',b} \ . 
\end{equation}
Such a scattering matrix distinguishes backscattering from other transitions.
 Other choices of $\VS^{(v)}$ include equi-transmitting matrices, which retain democratic scattering probabilities but prohibit backscattering  \cite{HSW07,KOR14}.

To quantize a binary graph we employ the $2\times 2$ DFT matrix, 
\begin{equation}
\VS^{(v)} = \frac{1}{\sqrt{2}} \left( 
\begin{array}{ c c }
1 & 1 \\
1 & -1 \\
\end{array}
\right)\ .
\label{DFTmatrix}
\end{equation}
The probability of scattering from $b$ to $b'$ is then,
\begin{equation}
\label{eq:democratic}
|\VS_{b',b}^{(v)}|^2 = 
\cases{2^{-1} & if  $t(b)=o(b')$ \\
0	& otherwise  \\} \ .
\end{equation} 
Note that, we can index the scattering matrices by the incoming and outgoing bonds, $\VS_{b',b}^{(v)}$, or by the terminus and origin vertices of these bonds, $\VS_{t(b'),o(b)}^{(v)}$.  

The \textit{bond scattering matrix} $\BS$ is a $B \times B$ matrix that collects all the vertex scattering amplitudes,
\begin{equation}
\label{bondscattmatrix}
\BS_{b',b} = \delta_{o(b'),v}\ \delta_{t(b),v}\   \VS_{b',b}^{(v)} \ . 
\end{equation}
We define a diagonal matrix of bond lengths $\ML=\diag\{ L_1,\dots,L_B \}$. 
The unitary matrix $\U(k)=\BS \rme^{\rmi k\ML }$ is the \textit{quantum evolution map}, see \cite{BK13,GS06}.
The one-parameter family of unitary matrices, $\BS \rme^{\rmi k\ML }$ over $k$, is a unitary-stochastic ensemble \cite{T01} where the associated stochastic matrix is the matrix of classical transition probabilities from bond $b$ to $b'$,
\begin{equation}
\label{transitionmatrix}
\T_{b',b} = |\BS_{b',b}|^2 \ . 
\end{equation}
We note that $\T$ is doubly stochastic, as each of the incoming and outgoing degrees at any vertex is $2$.
The $k$-spectrum of the quantum graph is the set of solutions of the \textit{secular equation},
\begin{equation}
\label{seculareqn}
\det(\BS \rme^{\rmi k \ML}-\UI)=0 \ .
\end{equation}
For a bond scattering matrix associated with a Laplace operator on the graph, if $k\neq 0$ is a solution of the secular equation, then $k^2$ is an eigenvalue of the Laplace operator with the same multiplicity.  According to the Weyl law, for a $k$-independent scattering matrix, the mean spacing of points in the $k$-spectrum is $2\pi/L_{\Gamma}$ where $L_{\Gamma}$ is the total length of the graph \cite{BK13}.

For a $B\times B$ matrix $\U(k)$, the \textit{characteristic polynomial} is,
\begin{equation} 
\label{charpoly matrix}
F_\zeta(k) = \det(\U (k) - \zeta \UI) = \sum\limits_{n=0}^{B} a_n(k) \zeta^{B-n} \ . 
\end{equation}
Note that, when $\U (k)$ is the quantum evolution map $\BS \rme^{\rmi k\ML }$ for a quantum graph, then $F_{1}(k)=0$ is the secular equation (\ref{seculareqn}). 
The unitarity of the quantum evolution map $\U (k)$ produces a Riemann-Siegel lookalike formula \cite{KS99}
associating pairs of coefficients of the characteristic polynomial,
\begin{equation}\label{eq:Riemann-Siegel}
a_n=\overline{a}_{B-n} \ .
\end{equation} 

The coefficients of the characteristic polynomial will be our main object of study.  It is well known that a quantum graph has an exact trace formula for the $k$-spectrum as a sum over periodic orbits on the graph \cite{KS99,BE09}.  Coefficients of the characteristic polynomial can be expressed as a sum over a finite set of pseudo orbits, collections of periodic orbits, where the pseudo orbits visit precisely $n$ bonds \cite{KS99,Aetal00,BHJ12}.  Hence all the spectral information is encoded in orbits of, at most, half the number of bonds of the graph. 

\subsection{Periodic Orbits and Pseudo Orbits}
\label{sec:orbits}

A \textit{path} of length $n$ is a sequence of bonds, 
\begin{equation*}
\label{pathbonds}
(b_1, b_2, \dots, b_n) \ , 
\end{equation*}
such that $t(b_i)=o(b_{i+1})$ for each $1\leq i\leq n-1$.  
Equivalently, we can define a path by a sequence of adjacent vertices,  
\begin{equation*}
\label{pathvertices}
(v_1, v_2, \dots, v_n, v_{n+1})\ , 
\end{equation*}
where $b_j=(v_j, v_{j+1})$.
A \textit{closed path} starts and ends on the same vertex, so that $t(b_n)=o(b_1)$ and $v_{n+1} = v_1$.  

On binary graphs with $V=2^r$, a path of length $n$ can be labeled by a word of length $n+r$,
\begin{equation*}
w_0 w_1 \cdots w_{n-1} w_n \cdots w_{n+r-1}\ ,
\end{equation*}
with $w_j\in \{ 0, 1\}$.  
The first $r+1$ entries label the bond $b_1$, and the consecutive groupings of $r+1$ adjacent letters label the sequence of bonds in the path.  

Closed paths of length $n$ are labeled by words of length $n$, 
\begin{equation*}
w_0 w_1 \cdots w_{n-1} \ ,
\end{equation*}
where the $n$ consecutive groupings of $r+1$ adjacent entries are read cyclically.
A \textit{periodic orbit} is an equivalence class of closed paths under cyclic permutations and
a \textit{primitive periodic orbit} is a periodic orbit that is not a repetition of a shorter periodic orbit.
We denote the number of primitive periodic orbits of length $n$ on a binary graph with $V$ vertices by $\PO_V(n)$.

To define \emph{lexicographic order}, consider a totally ordered alphabet $\mathcal{A}$ and two words,
\begin{eqnarray}
w&=a_1a_2\dots a_r \\
w'&=a_1'a_2'\dots a_s' 
\end{eqnarray}
with $a_i,a'_j \in \mathcal{A}$.  Then $w\triangleright w'$ if and only if there exists $i\leq \min \{r,s\} $ such that $a_1=a_1'$, $\dots$, $ a_{i-1}=a_{i-1}'$ and $a_i>a_i'$ or $r>s$ and $a_1=a_1'$, $\dots$, $a_s=a_s'$.
A \textit{Lyndon word} is a word that is strictly less than all of its cyclic permutations, in lexicographic order \cite{L83}.
So, for example, the binary Lyndon words of length $\leq 3$ arranged in lexicographic order are,
\begin{displaymath}
0\triangleleft 001 \triangleleft 01 \triangleleft 011 \triangleleft 1 \ .
\end{displaymath}
Thus, a Lyndon word can serve as a distinct representative of its equivalence class of closed paths.  
In addition, because a Lyndon word cannot be the same as any of its rotations, it labels a periodic orbit that is not a repetition of a shorter orbit.
Hence, primitive periodic orbits of length $n$ on a binary graph with $V=2^r$ vertices are in bijection with the set of all Lyndon words of the same length over the alphabet $\{ 0, 1\}$, independent of $r$.
So when $V=2^r$, we have $\PO_V(n) = \Lyndon$, the number of binary Lyndon words of length $n$.  We denote the set of all Lyndon words $\mathcal{L}$. 

For a periodic orbit $\gamma=(b_1, b_2, \dots, b_n)$, we denote the number of bonds in $\gamma$, the \textit{topological length} of $\gamma$, by $B_{\gamma}=n$, and the \textit{metric length} of $\gamma$ by, 
\begin{equation}
\label{pometriclength}
L_{\gamma} = \sum_{j=1}^n L_{b_j} \ . 
\end{equation}
A periodic orbit also has an associated \emph{stability amplitude}, the product of the scattering amplitudes around the orbit,
\begin{equation}
\label{postabamp}
A_{\gamma} =\BS_{b_1, b_n} 
\BS_{b_n, b_{n-1}}
\cdots
\BS_{b_3, b_2}
\BS_{b_2, b_1} \ . 
\end{equation}

A collection of periodic orbits is a \textit{pseudo orbit}, $\bg= \{\gamma_1, \gamma_2, \dots, \gamma_M\}$.  
A \textit{primitive pseudo orbit} is a pseudo orbit that contains only distinct primitive periodic orbits, so none of the primitive periodic orbits are repeated in the pseudo orbit.  Another class of pseudo orbits that will be significant are irreducible pseudo orbits \cite{BHJ12}.  An \emph{irreducible pseudo orbit} is a pseudo orbit where none of the bonds in the pseudo orbit are repeated.  Consequently all irreducible pseudo orbits are primitive pseudo orbits.    

Keeping with the periodic orbit notation, a pseudo orbit has a topological length,
\begin{equation}
\label{psotoplength}
B_{\bg} = \sum_{j=1}^M B_{\gamma_j} \ , 
\end{equation}
a metric length,
\begin{equation}
\label{psometriclength}
L_{\bg} = \sum_{j=1}^M L_{\gamma_j} \ , 
\end{equation}
and a stability amplitude,
\begin{equation}
\label{psostabamp}
A_{\bg} = \prod_{j=1}^M A_{\gamma_j} \ . 
\end{equation}
The number of periodic orbits in $\bg$ is denoted $m_{\bg}=M$.
The number of primitive pseudo orbits of length $n$ on a binary graph with $V$ vertices will be denoted $\PPO_V(n)$.

In \cite{BHJ12} the coefficients of the characteristic polynomial $F_\zeta(k)$ are written as a sum over primitive pseudo orbits where the topological length is the index of the coefficient.
\begin{theorem}\label{thm:coeffs}
	The coefficients of the characteristic polynomial $F_{\zeta}(k)$ are given by
	\begin{equation}
	a_n = 
	\sum_{\bg | B_{\bg}=n} 
	(-1)^{m_{\bg}} 
	A_{\bg} 
	\rme^{\rmi kL_{\bg}} \ ,
	\label{charpolycoeff}
	\end{equation} 
	where $\bg$ is a primitive pseudo orbit on the quantum graph.
\end{theorem}
In \cite{BHJ12} it is also shown that, following various cancellation mechanisms (see also \cite{Aetal00} for a related cancellation scheme), the sum can be further reduced to sum only over irreducible pseudo orbits $\bg$ with $B_{\bg}=n$.

The average of $a_n$ over $k$ is zero for $n\geq 1$.
Consequently, the first interesting statistic is the variance,
\begin{eqnarray}
\langle |a_n|^2 \rangle_k
&= \sum_{\bg, \bg' | B_{\bg} = B_{\bg'} = n} 
(-1)^{m_{\bg}+m_{\bg'}} 
A_{\bg} \bar{A}_{\bg'}   
\lim\limits_{K\to\infty} \frac{1}{K} \int_0^K
\rme^{\rmi k(L_{\bg}-L_{\bg'})}\ \rmd k  \\
&= \sum_{\bg, \bg' | B_{\bg} = B_{\bg'} = n} 
(-1)^{m_{\bg}+m_{\bg'}} 
A_{\bg} \bar{A}_{\bg'}   
\delta_{L_{\bg}, L_{\bg'}} \ .
\label{avgvarcoeff} 
\end{eqnarray}		
To contribute to the variance, a pair of primitive pseudo orbits $\bg, \bg'$ must have the same topological and metric lengths.
When the bond lengths are incommensurate, in order to contribute,  both pseudo orbits must visit each bond the same number of times.

To evaluate this sum, one must identify pairs of pseudo orbits with the same topological and metric lengths.  
The simplest way in which to obtain a pair of this type is to couple a primitive pseudo orbit with itself, which we call the \emph{diagonal contribution}.  (Binary quantum graphs lack time-reversal symmetry by construction, as the direction of the orbits cannot be reversed.)  
The diagonal contribution to (\ref{avgvarcoeff}) is,
\begin{equation}
\langle |a_n|^2 \rangle_{\diag} 
= \sum_{\bg | B_{\bg}= n}  |A_{\bg}|^2 
= \left( \frac{1}{2} \right)^n \PPO_{V}(n) \ ,
\label{diagpart}
\end{equation}
as $|A_{\bg}|^2 = 2^{-n}$ using (\ref{eq:democratic}).  This was evaluated for $q$-nary quantum graphs with $V=q^r$ in \cite{BHS19}.

\section{The Variance of the Characteristic Polynomial's Coefficients}
\label{sec:variance}

\subsection{Self-Intersections}
\label{sec:enc}

For two different pseudo orbits to visit the same bonds the same number of times we must be able to reorder the bonds in one pseudo orbit to produce the other.  This requires a pseudo orbit that has one or more self-intersections.
Given a primitive pseudo orbit $\bg = \{ \gamma_1, \gamma_2, \dots, \gamma_{m_{\bg}} \}$, each $\gamma_j$ is a primitive periodic orbit that can be written as a sequence of vertices and bonds, for $1 \leq j \leq m_{\bg}$.
A \emph{self-intersection}, or an \emph{$\ell$-encounter}, is a repeated subsequence of vertices and/or bonds $\enc = (v_0, v_1, \dots, v_{\tilde{n}-1}, v_{\tilde{n}})$ of maximal length that appears exactly $\ell$ times in the primitive pseudo orbit.  Such a repeated sequence can occur inside a single periodic orbit or it can be a sequence appearing in multiple orbits or a combination of these.
The vertices $s_1, s_2, \dots s_{\ell}$ immediately preceding $v_0$ and the vertices $f_1, f_2, \dots, f_{\ell}$ immediately following $v_{\tilde{n}}$ for each of the repetitions are distinct for some pair, so $s_i \neq s_j$ for some $1 \leq i < j \leq \ell$ and $f_{i'} \neq f_{j'}$ for some $1 \leq i' < j' \leq \ell$.  Equivalently the $s_j$ are not all equal and the $f_j$ are not all equal.  Note that, for a binary graph, there are only two distinct choices of incoming vertices adjacent to $v_0$ (and similarly for the outgoing vertices).  
The \emph{encounter length} is the number of bonds $\tilde{n}$ in the self-intersection; if $\tilde{n} = 0$, then the $\ell$-encounter is a single vertex, an encounter of length zero.
The sequences of vertices and bonds that connect encounters 
are called \emph{links}.

We will classify primitive pseudo orbits by the number of self-intersections, the number of repetitions of each self-intersection, and the lengths of the self-intersections. To illustrate the structure of pseudo orbits with self-intersections, we consider a few examples that we use subsequently.

\subsubsection{A Single $2$-Encounter}
\label{sec:one 2-encounter}
If $\bg$ contains a single self-intersection, then the periodic orbits not containing the encounter sequence cannot intersect any other orbits.
If a primitive pseudo orbit $\bg$ contains a single $2$-encounter, then the vertices preceding $v_0$ are distinct, $s_1 \neq s_2$, as are the vertices $f_1 \neq f_2$ following $v_{\tilde{n}}$.
Also no subsequence of the $2$-encounter is repeated three or more times.  Moreover, the links must differ; otherwise, $\bg$ is not primitive.  Figure \ref{figfigure8enc} shows a subgraph from which we construct examples of primitive pseudo orbits with a single $2$-encounter.

\begin{figure}
	\begin{center}
	\includegraphics[scale=1, trim=145 510 125 125, clip]{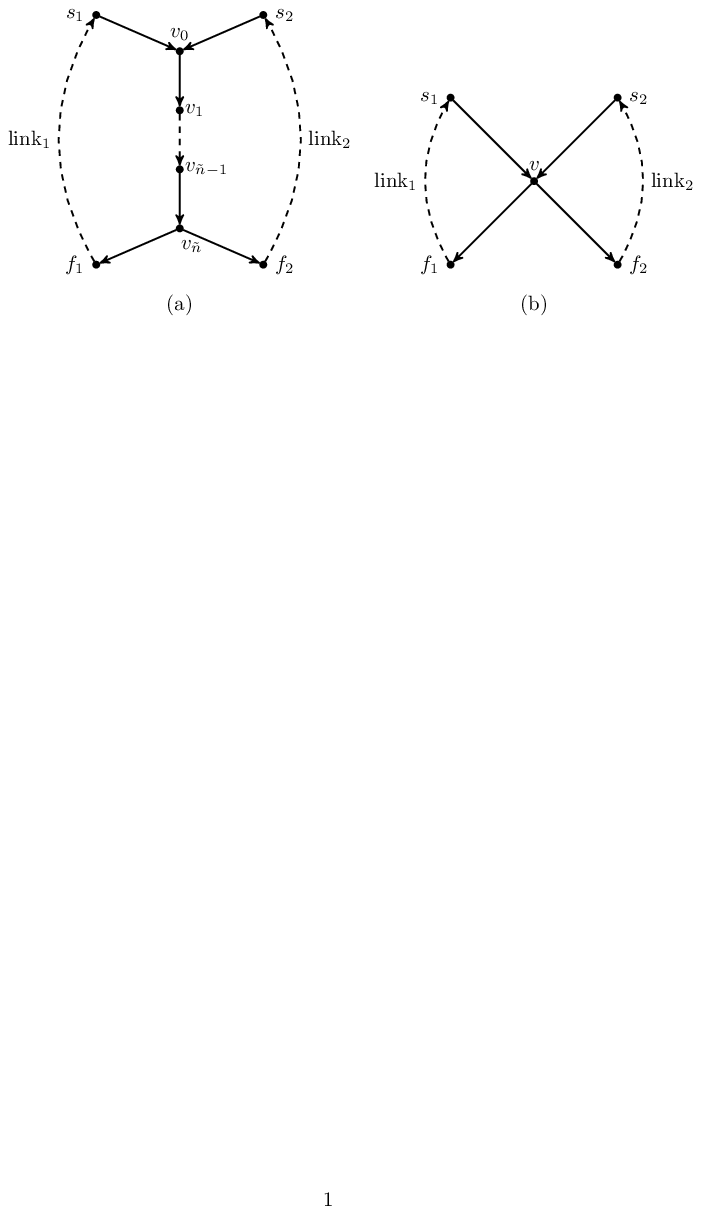}
	\caption{Subgraphs from which primitive pseudo orbits containing a figure eight orbit can be constructed. Such primitive pseudo orbits contain a self-intersection (a) along $\tilde{n} > 0$ bonds or (b) where the intersection is a single vertex $v$, so $\tilde{n} = 0$.} \label{figfigure8enc}
	\end{center}	
\end{figure}

	Let $\bg = \{ \gamma_1, \gamma_2, \dots, \gamma_{m_{\bg}} \} $ be a primitive pseudo orbit where $\gamma_1$ is a closed sequence of vertices and bonds shown in the subgraph of figure \ref{figfigure8enc}(a) such that
	\begin{equation}
	\gamma_1 =  (s_1, \enc, f_2, \dots, s_2, \enc, f_1, \dots, s_1)  \ .
	\end{equation}
	The primitive periodic orbits $\gamma_2, \dots, \gamma_{m_{\bg}}$ do not contain any vertices in $\gamma_1$, and they collectively contain no vertex more than once.
	The periodic orbit $\gamma_1$ is often called a figure eight orbit, as introduced in \cite{S02,SR01}.  
	Then the vertex sequence $\enc$, which consists of the bonds $(v_0, v_1)$, $(v_1, v_2), \dots,$ $(v_{\tilde{n}-1}, v_{\tilde{n}})$ is a $2$-encounter.
	The self-intersection is entered from each of the two distinct vertices $s_1, s_2$; after entering from $s_1$ (or $s_2$) the orbit exits the self-intersection to vertex $f_2$ (or $f_1$, respectively) with each exit vertex $f_1, f_2$ distinct.
	We refer to the remaining sequences as links; in \ref{figfigure8enc}(a), for $i=1,2$,
	\begin{equation}
	\link_i = (v_{\tilde{n}}, f_i, \dots, s_i, v_0) \ . 
	\end{equation}

The primitive pseudo orbit $\bg' = \{ \gamma', \gamma'', \gamma_2, \dots, \gamma_{m_{\bg}}  \}$ where 
\begin{eqnarray}
\gamma' &= (s_1, \enc, f_1, \dots, s_1) \ , \label{eq:gamma'}\\
\gamma'' &= (s_2, \enc, f_2, \dots, s_2) \ , \label{eq:gamma''}
\end{eqnarray}
has the same topological and metric lengths as $\bg$, as each of $\gamma'$ and $\gamma''$ contains one link of the figure eight in $\gamma_1$ and one traversal of the encounter sequence.
Thus in both $\bg$ and $\bg'$ all bonds are used the same number of times and the encounter sequence is used twice in total.
Note that, this is the only way to reorder the sequence in such a way as to pair $\bg$ with a partner $\bg' \neq \bg$ of the same topological and metric lengths, and that $m_{\bg'} = m_{\bg} + 1$, as we have split one orbit in $\bg$ into two orbits in $\bg'$.

Similarly, we can reverse the roles of $\bg$ and $\bg'$ so the pseudo orbit $\bg$ contains the vertex sequence $\enc$ as a $2$-encounter with 
\begin{eqnarray}
\gamma_1 &= (s_1, \enc, f_1, \dots, s_1) \ , \\
\gamma_2 &= (s_2, \enc, f_2, \dots, s_2) \ ,
\end{eqnarray}
such that the primitive periodic orbits $\gamma_3, \dots, \gamma_{m_{\bg}}$ do not contain any of the vertices in $\gamma_1$ or $\gamma_2$, and they collectively contain no vertex more than once.
Then the only way to pair $\bg$ with a primitive pseudo orbit $\bg' \neq \bg$ such that $B_{\bg'} = B_{\bg}$ and $L_{\bg'} = L_{\bg}$ is to join $\gamma_1, \gamma_2$ at the self-intersection and obtain
\begin{equation}
\gamma =  (s_1, \enc, f_2, \dots, s_2, \enc, f_1, \dots, s_1)  \ ,
\end{equation} 
with $\bg' = \{ \gamma, \gamma_3, \dots, \gamma_{m_{\bg}} \}$.
Note $m_{\bg'} = m_{\bg}-1$.

Both of the previous examples also make sense when the encounter occurs at a single vertex.
We replace the encounter sequence of vertices $\enc$ with the single vertex $v$ and use the previous orbit subsequences as they appear in figure \ref{figfigure8enc}(b).  Then these are examples of $2$-encounters of length zero, $\tilde{n} = 0$.

\subsubsection{A Single $\ell$-Encounter with Distinct Links}\label{sec:single encounter distinct links}
Now we consider a single $\ell$-encounter with $\ell \geq 2 $  where we require that no two link sequences are the same.
For a binary graph there are only two incoming and two outgoing bonds at each vertex, so there will necessarily be repeated vertices adjacent to the self-intersection when $\ell \geq 3$; we do not regard these as part of the self-intersection as they are not repeated the maximum number of times $\ell$.

Let $\bg = \{ \gamma_1, \gamma_2, \dots, \gamma_{m_{\bg}} \} $ be a primitive pseudo orbit where the primitive periodic orbit $\gamma_1$ is given by
\begin{equation}
\gamma_1 = (s_1, \enc, f_2, f_2' \dots, s_2', s_2, \enc, f_2, f_2'', \dots, s_2'', s_2, \enc, f_1, \dots, s_1) \ ,
\label{3enc123}
\end{equation}
see figure \ref{figbinarylenc}.
The bond $(s_2, v_0)$ appears twice at the beginning of the encounter sequence and the bond $(v_{\tilde{n}}, f_2)$ appears twice at the end of the encounter sequence.  
This is necessitated by the structure of a binary graph, and we do not count either bond as part of the 3-encounter as neither bond is repeated three times.
The links in this example are,
\begin{eqnarray}
\link_1 &= (v_{\tilde{n}}, f_1, \dots, s_1, v_0) \ , \\
\link_2 &= (v_{\tilde{n}}, f_2, f_2', \dots, s_2', s_2, v_0) \ , \\
\link_3 &= (v_{\tilde{n}}, f_2, f_2'', \dots, s_2'', s_2, v_0) \ .
\end{eqnarray}

\begin{figure}[htb!]
	\centering
	\includegraphics[scale=1, trim=100 450 175 125, clip]{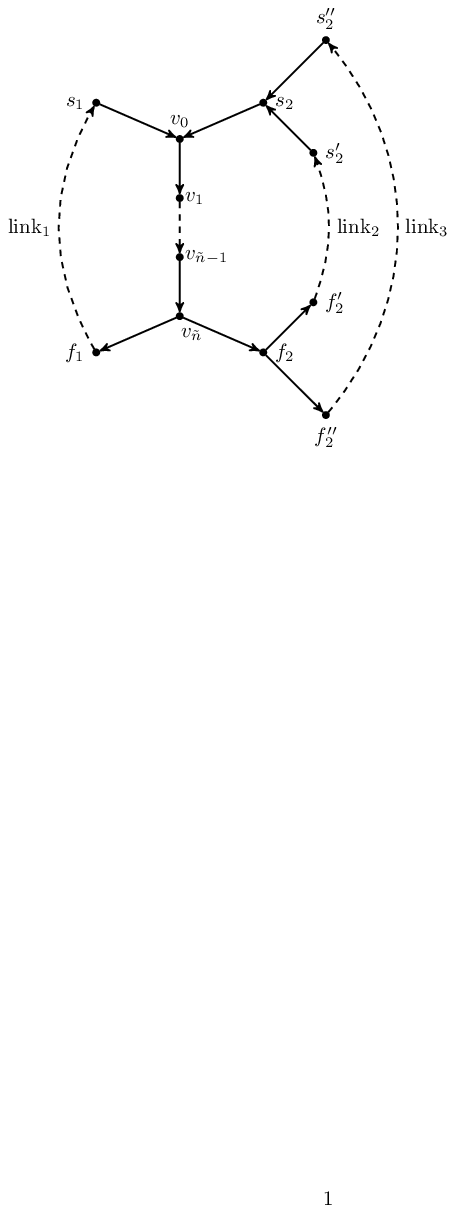}
	\caption[A subgraph from which a primitive pseudo orbit with a $3$-encounter with three distinct links can be constructed]{A subgraph from which a primitive pseudo orbit with a $3$-encounter with three distinct links can be constructed.} 
	\label{figbinarylenc}
\end{figure}

There are six distinct partner pseudo orbits for $\bg$ that have the same topological and metric lengths corresponding to the elements of the permutation group $S_3$.
As a traversal of a link sequence is always followed by traversal of the encounter, we denote a periodic orbit as a sequence of links and assume that the encounter sequence follows each link.
So $\gamma_1$ defined in (\ref{3enc123}) is equivalently denoted $\gamma_1 = (\link_1, \link_2, \link_3) $ and corresponds to the permutation $(1 \ 2 \ 3)$.  We let $\sigma_{\bg}$ denote the permutation corresponding to a primitive pseudo orbit $\bg$ with a single self-intersection and no repeated links.
The partner pseudo orbits of $\bg$ are listed in table \ref{figExTable3enc}, with $e$ the identity permutation.

\begin{table}[htb!]
	\caption[Primitive periodic orbits that can replace $\gamma_1$ to produce a partner pseudo orbit $\bg'$ in example ??.]{All primitive periodic orbits that can replace $\gamma_1 = (\link_1, \link_2, \link_3)$ to produce a partner pseudo orbit $\bg'$ of the same metric length}
	\centering
	\begin{tabular}{ c  c  c }
		\hline
		Orbit(s) replacing $\gamma_1$ in $\bg'$ & $\sigma_{\bg'}$ & $m_{\bg'}$      \\
		\hline
		None & $(1 \ 2 \ 3)$ & $m_{\bg'}=m_{\bg}$ \\
		$\gamma = (\link_1, \link_3, \link_2)$ & $(1 \ 3 \ 2)$ & $m_{\bg'}=m_{\bg}$ \\
		$\gamma' = (\link_1, \link_2), \gamma'' = (\link_3)$ & $(1 \ 2)$ & $m_{\bg'}=m_{\bg} + 1$ \\
		$\gamma' = (\link_1, \link_3), \gamma'' = (\link_2)$ & $(1 \ 3)$ & $m_{\bg'}=m_{\bg} + 1$ \\
		$\gamma' = (\link_2, \link_3), \gamma'' = (\link_1)$ & $(2 \ 3)$ & $m_{\bg'}=m_{\bg} + 1$ \\
		$\gamma' = (\link_1), \gamma'' = (\link_2), \gamma''' = (\link_3)$ & $e$ & $m_{\bg'}=m_{\bg} + 2$ \\
		\hline
	\end{tabular}
	\label{figExTable3enc}
\end{table}

In section \ref{sec:one 2-encounter} the partner orbits corresponded to permutations in $S_2$.  The permutation $(1 \ 2)$ corresponds to $\gamma_1 \in \bg$ and its partners are $\bg' = \bg$ and the pseudo orbit where $\gamma_1$ is replaced with $\gamma', \gamma''$, defined in (\ref{eq:gamma'}) and (\ref{eq:gamma''}), where $\bg'$ corresponds to the identity permutation $e$.
Similarly, the partner pseudo orbits of a primitive pseudo orbit $\bg$ on a binary graph with a single $\ell$-encounter and $\ell$ distinct links correspond to elements of the permutation group $S_\ell$.  

\subsubsection{A Single $\ell$-Encounter with Repeated Links}
\label{sec:single lenc repeated links}
Let $\bg$ be a primitive pseudo orbit containing a single $\ell$-encounter for $\ell\geq 3$, where one or more links are repeated.
Here the partner pseudo orbits correspond to Lyndon tuples over a multiset of link indices, a generalization of permutations that relies on the Chen-Fox-Lyndon theorem \cite{BHS19,CFL58,D83,F19-1,F19-2,L83}.

\begin{theorem}[Chen-Fox-Lyndon]
	Every non-empty word $w$ can be uniquely formed by concatenating a non-increasing sequence of Lyndon words in lexicographic order, the Lyndon decomposition of $w$.
	So
	\begin{equation}
	w = l_k  \dots l_2 l_1 \ , 
	\end{equation}
	where $l_1, \dots, l_k \in \mathcal{L}$ and $l_{i} \trianglerighteq l_{i-1}$, for $i=2,\dots, k$.
	\label{thmCFL}
\end{theorem}

A multiset \cite{F19-1,F19-2}, over the finite ordered alphabet $\mathcal{A} =  \{ 1, 2, \dots, \mu \}$ is denoted $\cM = [ 1^{m_1}, 2^{m_2}, \dots, \mu^{m_{\mu}} ]$ where $m_a \in \mathbb{N}_0$ is the multiplicity of $a$ in $\cM$ and the 
cardinality of $\cM$ is $\sum_{a \in \mathcal{A}} m_a = \ell$.
Let $w = a_1 a_2 \dots a_{\ell}$ be an $\ell$-word that uses each element of $\cM$ once.
By theorem \ref{thmCFL}, $w$ has a unique Lyndon decomposition $w = l_k \dots l_2 l_1$.  
For those words with a strictly decreasing factorization $l_k \triangleright \cdots \triangleright l_2 \triangleright l_1$, the \emph{Lyndon tuple} of $w$ over $\cM$ is $\tup(w) = (l_1, l_2, \dots, l_k)$.
We denote the set of all Lyndon tuples over $\cM$ by $\mathcal{L}(\cM)$.

The \emph{weight} of a letter $a \in \mathcal{A}$ is a formal variable $x_a$.
For any word $w = a_1 a_2 \dots a_\ell$, the \emph{weight} $\wt(w) = x_{a_1} \dots x_{a_\ell}$ is the product of the weights of its letters.
The \emph{weight} of the Lyndon tuple $\tup(w) = (l_1, l_2, \dots, l_k)$ is the product of the weights of its Lyndon words, 
\begin{equation}
\wt(w) = \prod_{i=1}^k \wt(l_i) \ .
\end{equation}

For any non-empty $\ell$-word $w$ with $\tup(w) = (l_1, l_2, \dots, l_k) \in \mathcal{L}(\cM)$, we define the \emph{Lyndon index} of $w$ to be the number $i_{\mathcal{L}}(w) = \ell-k$.
A non-empty $\ell$-word $w$ with $\tup(w) = (l_1, l_2, \dots, l_k) \in \mathcal{L}(\cM)$ is \emph{even (odd)} if its Lyndon index $i_{\mathcal{L}}(w)$ is even (odd).
Note that, we are only assigning a notion of parity to $\ell$-words that have a strictly decreasing Lyndon word decomposition.
We let the set of $\ell$-words for which the corresponding tuple is even (odd) be denoted by $E$ ($O$).
The following theorem of Faal \cite{F19-2}, a generalization of a theorem of Sherman \cite{S60}, implies that there are the same number of odd and even Lyndon tuples in $\mathcal{L}(\cM)$ for any multiset $\cM$.

\begin{theorem}
	\label{thm:splitting}
	Let $\mathcal{A} = \{ 1, \dots, \mu \}$ be a finite ordered alphabet with weights $\{ x_1,\dots , x_\mu\}$ and fix a multiset $\cM = [ 1^{m_1}, 2^{m_2}, \dots, \mu^{m_{\mu}} ]$ over $\mathcal{A}$ of cardinality $\ell > 1$.
	The weighted sum of even $\ell$-words over $\cM$ is the same as the weighted sum of odd $\ell$-words over $\cM$, 
	\begin{equation}
	\sum_{w \in E} \wt(w) = \sum_{w \in O} \wt(w) \ .
	\end{equation}
\end{theorem}

The proof of theorem \ref{thm:splitting} relies on a weight-preserving bijection $f$ between odd and even Lyndon words, which we make use of later.  The bijection $f$ depends on whether the first word in the Lyndon decomposition is splittable, which is itself determined by a standard factorization of the first Lyndon word \cite{L83}.

If $w$ is a Lyndon word that is not a single letter, and $w=rs$ such that $r$ and $s$ are Lyndon words and $s$ is of maximal length, then the pair $(r,s)$ is the \emph{standard factorization} of $w$.
\begin{proposition} 
	For a Lyndon word $w \in \mathcal{L}\backslash\mathcal{A}$, with $w=rs$ for $r,s \in \mathcal{L}$, then $(r,s)$ is the standard factorization of $w$ if and only if $s$ is the smallest proper suffix of $w$ lexicographically.
	\label{propstdfactsuffix}
\end{proposition}

The first Lyndon word $l_1$ in $\tup(w) = (l_1, l_2, \dots, l_k)$ is \emph{splittable} if it is not a single letter and its standard factorization $(r_1, s_1)$ satisfies $s_1 \triangleleft l_2$.
If the first Lyndon word $l_1\in \tup(w)$ is splittable, we can define $f_1(\tup(w)) = (r_1, s_1, l_2, \dots, l_k)$ where $(r_1, s_1)$ is the standard factorization of $l_1$.  If the first Lyndon word $l_1\in \tup(w)$ is not splittable, we can define $f_2(\tup(w)) = (l_0, l_3, \dots, l_k)$ where $l_0 = l_1 l_2$.
Then an odd (even) $\ell$-word uniquely maps to an even (odd) $\ell$-word by
\begin{equation}
f(w) = \cases{
w_1 & if  $\tup(w_1) = f_1(\tup(w))$ when $l_1$ is splittable,  \\
w_2 & if  $\tup(w_2) = f_2(\tup(w))$ when  $l_1$ is not splittable.}
\label{eqPPOmaptuples}
\end{equation}

Like a pseudo orbit with a single self intersection and distinct links, where we label partner pseudo orbits with permutations, we can label partner pseudo orbits where there are repeated links, with Lyndon tuples where each letter labels a link in the orbit.  The multiset $\cM$ records the number of times each link appears in $\bg$ and $\mathcal{L}(\cM)$ is the set of possible partner pseudo orbits where, for $\tup(w)\in \mathcal{L}(\cM)$, each Lyndon word corresponds to a unique primitive periodic orbit, where each link is followed by the encounter sequence.  The Lyndon word is not only a unique label of the equivalence class of closed paths, but must label a primitive periodic orbit since it is not a repetition of a shorter word.
Thus, theorem \ref{thm:splitting} implies that the number of primitive partner pseudo orbits which contain an even number of periodic orbits is equal to the number of primitive partner pseudo orbits that contain an odd number of periodic orbits.  This is central to the parity argument in section \ref{sec:contributions}.

\subsubsection{Multiple Encounters}
\label{sec:mult enc}

We now consider a primitive pseudo orbit $\bg $ with $N$ self-intersections.
A sequence of adjacent vertices $v_0^i, \dots, v_{\tilde{n}_i}^i$ is a \emph{self-intersection} or an \emph{$\ell_i$-encounter} if it appears exactly $\ell_i$ times in the primitive periodic orbit sequences of $\bg$ such that the vertices $s_1^i, \dots s_{\ell_i}^i$ immediately preceding $v_0^i$ and the vertices $f_1^i, \dots, f_{\ell_i}^i$ immediately following $v_{\tilde{n}_i}^i$ for each repetition are distinct for some pair.  So $s_h^i \neq s_j^i$ for some $1 \leq h< j \leq \ell_i$ and $f_{h'}^i \neq f_{j'}^i$ for some $1 \leq h' < j' \leq \ell_i$, for all $i=1, \dots, N$.
Then $\bg$ has $N$ self-intersections of types $\vec{\ell} = (\ell_1, \ell_2, \dots, \ell_N)$.  We denote the set of primitive pseudo orbits of length $n$ with $N$ self intersections $\mathcal{P}^n_N$.

At the $i$-th encounter there are $\ell_i$ sequences of vertices that begin at $v^i_{\tilde{n}_i}$ and each of these ends at $v_0^j$ for some $j = 1, \dots, N$; there are also $\ell_i$ sequences of vertices that end at $v_0^i$ and begin at $v_{\tilde{n}_j}^j$ for some $j=1,\dots, N$.
These sequences do not contain encounter sequences, and we refer to them as \emph{outgoing} and \emph{incoming links} at the encounter, respectively.  
In the case of a single encounter, each link is both incoming and outgoing to the encounter. For multiple encounters, a link sequence is incoming to one encounter and outgoing at one encounter, which may be different from the encounter at which it is incoming.
If $\ell_i = 2$ for some encounter, then the two incoming links must be distinct and similarly the two outgoing links are distinct.
However, if $\ell_i \geq 3$ for an encounter, then there will be at least two distinct incoming links and at least two distinct outgoing links, but links may be used more than once.

A primitive pseudo orbit with multiple self-intersections could have several sets of periodic orbits like those described in sections \ref{sec:one 2-encounter} through \ref{sec:single lenc repeated links} where the sets do not overlap one another.  
It is also possible to describe primitive pseudo orbits with multiple self-intersections on a subgraph where the encounters behave topologically as vertices and the links as topological edges, see figure \ref{figmult}. 
\begin{figure}
	\centering
	\includegraphics[scale=1, trim=150 500 100 125, clip]{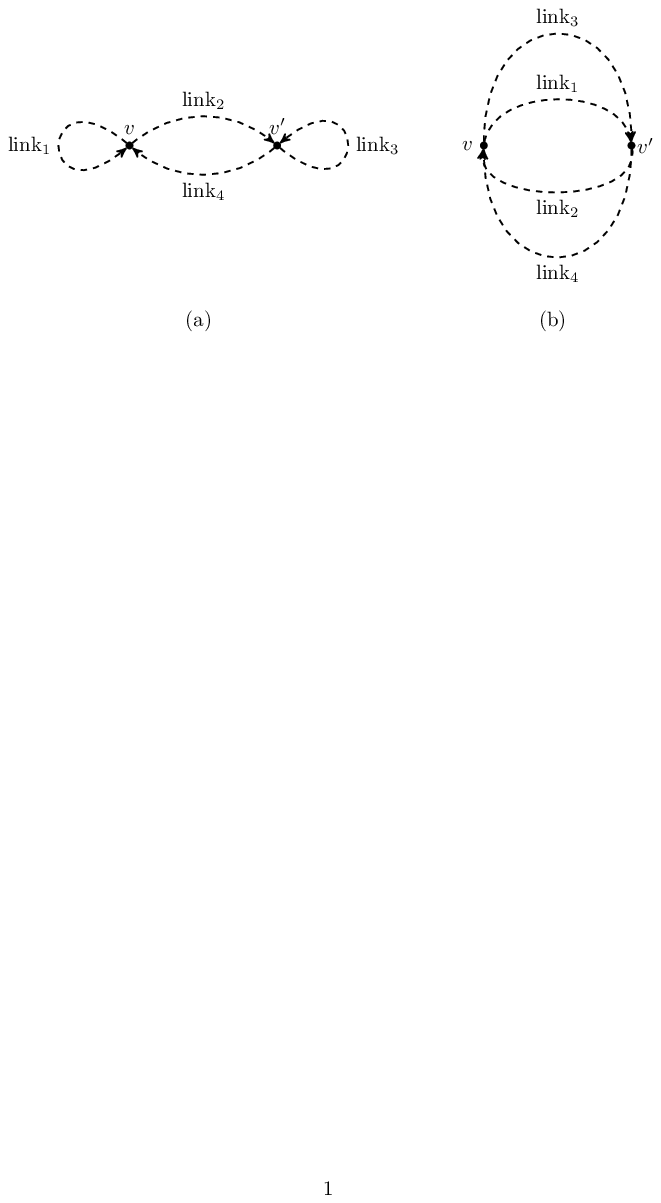}
	\caption[Schematic representations of subgraphs from which primitive pseudo orbits with multiple encounters can be constructed]{Schematic representations of subgraphs from which primitive pseudo orbits with multiple encounters can be constructed.} 
	\label{figmult}
\end{figure}

Let $\bg = \{ \gamma_1, \gamma_2, \dots, \gamma_{m_{\bg}} \} $ where $\gamma_1 = (\link_1, \link_2, \link_3, \link_4)$ is an orbit on the subgraph shown schematically in figure \ref{figmult}(a), and no other orbit in $\bg$ has self-intersections.
Then $\bg$ has two $2$-encounters of length zero at the vertices $v$ and $v'$.
Note that $\link_1$ is both incoming and outgoing at the $2$-encounter $v$, while $\link_2$ is outgoing at $v$ and incoming at the $2$-encounter $v'$; $\link_4$ is outgoing at $v'$ and incoming at $v$ while $\link_3$ is incoming and outgoing at $v'$.

Let $\bg = \{ \gamma_1, \gamma_2, \dots, \gamma_{m_{\bg}} \} $ where $\gamma_1 = (\link_1, \link_2, \link_3, \link_3, \link_4)$ is an orbit on the subgraph shown schematically in figure \ref{figmult}(a), and no other orbit in $\bg$ has encounters.
Then $\bg$ has a $2$-encounter of length zero at the vertex $v$ and a $3$-encounter of length zero at the vertex $v'$.

Let $\bg = \{ \gamma_1, \gamma_2, \dots, \gamma_{m_{\bg}} \} $ where $\gamma_1 = (\link_1, \link_2, \link_3, \link_4)$ is an orbit on the subgraph shown schematically in figure \ref{figmult}(b), and no other orbit in $\bg$ repeats vertices contained in $\gamma_1$ or has encounters.
Then $\bg$ has two $2$-encounters of length zero at the vertices $v$ and $v'$.

When a primitive pseudo orbit contained a single encounter, we generated its partner pseudo orbits by reordering links.  
Partner primitive pseudo orbits for a pseudo orbit with multiple encounters are primitive pseudo orbits where the links have again been reordered, although this reordering must respect the connectivity of the encounters, so outgoing links follow incoming links at an encounter.  In section \ref{sec:contributions mult enc various} we will see that we can extend a parity argument we develop for pseudo orbits with a single encounter to evaluate a sum over all partner pseudo orbits without explicitly formulating all reorderings.

\subsection{Alternative Formulation of the Variance}
\label{sec:alternative variance}

The variance of the coefficients of the characteristic polynomial of a quantum graph can be written as a finite sum over primitive pseudo orbit pairs where both orbits in the pair have the same topological length $n$ and the same metric length (\ref{avgvarcoeff}).
A primitive pseudo orbit $\bg$ of topological length $n$ has a well-defined number of self-intersections $0 \leq N \leq n$.
The upper bound $n$ is trivial as no pseudo orbit could have more self-intersections than there are vertices in the orbit.

Let $\mathcal{P}_{\bg}$ be the set of all primitive pseudo orbits $\bg'$ such that $B_{\bg'} = B_{\bg}$ and $L_{\bg'} = L_{\bg}$.
We define a contribution for primitive pseudo orbits that partner $\bg$,
\begin{equation}
C_{\bg} = \sum_{ \bg' \in \mathcal{P}_{\bg} }
(-1)^{m_{\bg}+m_{\bg'}}
A_{\bg} \bar{A}_{\bg'} \ .
\label{contributionC}
\end{equation}
Then the variance of the coefficients can be written as,
\begin{equation}
\label{variancewithdiag}
\langle |a_n|^2 \rangle_k
= 
\sum_{N=0}^n \sum_{\bg \in \mathcal{P}_N^n} C_{\bg} \ .
\end{equation}

\subsection{Evaluation Using Irreducible Pseudo Orbits}
\label{sec:irreducible variance}

The following argument, that produces proposition \ref{thm:main}, applies to all $4$-regular directed graphs with FFT scattering matrices and was described briefly in \cite{HH22}.   
It starts from a reduced form of theorem \ref{thm:coeffs}, obtained in \cite{BHJ12}, where the sum is over only the irreducible pseudo orbits with $n$ bonds.  
An irreducible pseudo orbit is a pseudo orbit in which no bonds are repeated.  Writing the coefficients as a sum over irreducible pseudo orbits produces a corresponding formula for the variance,
\begin{equation}
\langle |a_n|^2 \rangle_k
= \sum_{N=0}^n \sum_{\bg \in \hP_N^n} C_{\bg} \ ,
\label{avgvarcoeff2} 
\end{equation}	
where $\hP_N^n$ is the set of irreducible pseudo orbits of topological length $n$ with $N$ self-intersections.

If there are no self-intersections in $\bg$  (so $N=0$) the bonds in $\bg$ cannot be rearranged to produce partner orbits of the same length.  So $\bg$ has only a single partner pseudo orbit $\bg'=\bg$.  Then $C_{\bg}=|A_{\bg}|^2=1/2^n$.  Hence the $N=0$ contribution to the variance is $|\hP_0^n|/2^n$.  As every pseudo orbit with repeated bonds must have a self-intersection, $\hP_0^n= \mathcal{P}^n_0$ which explains the first term in (\ref{variancefinal}).

In an irreducible pseudo orbit with an encounter the encounter length must be zero, as all bonds in an encounter are repeated.  For a $4$-regular directed graph where each vertex has two incoming and two outgoing bonds, $\ell$-encounters with $\ell> 2$ also require repeating bonds in the pseudo orbit as there are only two bonds from which one can enter (or leave) the encounter.  Examples of such pseudo orbits were constructed in section \ref{sec:single encounter distinct links}.  So $\hP_N^n$ is the subset of $\mathcal{P}_N^n$ containing only primitive pseudo orbits where all $N$ self-intersections are $2$-encounters of length zero.  

It remains to evaluate $C_{\bg}$ for an irreducible pseudo orbit with self-intersections.  
Consider a primitive pseudo orbit $\bg\in \hP_{1}^n$, so $\bg$ has a single $2$-encounter of length zero. Schematically, the two pseudo orbits with a self-intersection that fit this scenario are the figure eight examples in section \ref{sec:one 2-encounter}. 
Let $A_{\link_i}$ for $i=1,2$ denote the product of scattering amplitudes at the interior vertices of $\link_i$.
If we take the pair of irreducible pseudo orbits,
\begin{eqnarray}
\bg&=\{ \gamma_1=(1\, 2), \gamma_2,\dots, \gamma_{m_{\bg}}\} \label{eq:single encounter 1}\\
\bg'&=\{\gamma'=(1), \gamma''=(2),\gamma_2,\dots, \gamma_{m_{\bg}}\}
\label{eq:single encounter 2}
\end{eqnarray}
then
\begin{equation}
A_{\gamma_1} = A_{\link_1}  A_{\link_2} \VS_{f_2,s_1}^{(v)} \VS_{f_1,s_2}^{(v)} \ .
\label{stabamp1,t=0}
\end{equation}
Reordering the links of $\bg$ at the encounter vertex $v$ to obtain $\bg'$ splits $\gamma_1$ into two orbits $\gamma'$ and $\gamma''$, and changes the scattering coefficients at $v$, 
\begin{equation} 
A_{\gamma'} A_{\gamma''}  = A_{\link_1}  A_{\link_2} \VS_{f_1,s_1}^{(v)} \VS_{f_2,s_2}^{(v)} \ .
\label{stabamp2,t=0}
\end{equation}
All other orbits in $\bg$ are unchanged in $\bg'$, and therefore,
\begin{equation}
A_{\bg} \bar{A}_{\bg'} 
= \left( \prod_{h=2}^{m_{\bg}} |A_{\gamma_h}|^2 \right) 
|A_{\link_1}|^2 |A_{\link_2}|^2 
\VS_{f_2,s_1}^{(v)} \VS_{f_1,s_2}^{(v)} \bar{\VS}_{f_1,s_1}^{(v)} \bar{\VS}_{f_2,s_2}^{(v)} \ .
\label{eq:stability 8}
\end{equation}
All four elements of the scattering matrix $\VS^{(v)}$ appear in the product and hence $A_{\bg} \bar{A}_{\bg'} = -2^{-n}$.
Similarly, when the roles of $\bg$ and $\bg'$ are reversed $A_{\bg} \bar{A}_{\bg'} = -2^{-n}$. 

Now consider an irreducible pseudo orbit $\bg\in \hP_{N}$ with $N$ self-intersections.
There is an associated vector $(\rho_1,\dots,\rho_N)$ of permutations $\rho_i \in S_2$ where $\rho_i$ records the order in which the two outgoing links follow the two incoming links at the $i$-th encounter in $\bg$.
Partner pseudo orbits $\bg'$ are generated by considering all possible vectors $(\rho'_1,\dots,\rho'_N)$  with $\rho_i' \in S_2$.
At the $i$-th self-intersection, if $\rho_i' = \rho_i$, then the product of scattering coefficients at that encounter is $1/4$ as two scattering coefficients are each used twice.
If, however, $\rho_i' \neq \rho_i$, each of the scattering coefficients at the encounter is used exactly once and their product is $-1/4$. 
Thus, if $j$ is the number of $2$-encounters at which the order of links in $\bg'$ is changed relative to $\bg$, then
\begin{equation}
\label{2encstability}
A_{\bg} \bar{A}_{\bg'}  
= (-1)^j\,  2^{-n} \ .
\end{equation}

The other source of sign changes in the formula (\ref{contributionC}) for $C_{\bg}$ is the factor $(-1)^{m_{\bg '}}$.
Consider again an irreducible pseudo orbit $\bg$ with a single $2$-encounter of length zero.
For the pseudo orbit  $\bg$ defined in (\ref{eq:single encounter 1}), reordering links at the self-intersection increases the number of periodic orbits in $\bg'$ relative to $\bg$; so $m_{\bg'} = m_{\bg} +1$.
Reversing the roles of $\bg$ and $\bg'$, reordering the links at the single self-intersection reduces the number of periodic orbits in $\bg'$ relative to $\bg$ by one.  Reordering links at a single self-intersection changes the parity of the number of orbits.
Hence if $j$ is the number of $2$-encounters at which $\bg$ is reordered to obtain the partner orbit $\bg'$ then, 
\begin{equation}
\label{2encparityoforbits}
(-1)^{m_{\bg}+m_{\bg'}} = (-1)^j \ .
\end{equation}

\begin{lemma}
	\label{thm:ContMult2enc}
	For an irreducible pseudo orbit $\bg$ of length $n$ with $N \geq 1$ self-intersections, $C_{\bg} = 2^{N-n}$.
\end{lemma}
\begin{proof}
	Let $\bg\in \hP_{N}^n$ be an irreducible pseudo orbit with $N \geq 1$ self-intersections and consider all partner orbits $\bg' \in \mathcal{P}_{\bg}$.
	Order the sum over $\mathcal{P}_{\bg}$ by the number of self-intersections at which link connections are reordered relative to $\bg$, which we denote by $j$.  Using (\ref{2encstability}) and (\ref{2encparityoforbits}),
	\begin{equation}
	C_{\bg} = \sum_{ \bg' \in \mathcal{P}_{\bg}} 
	(-1)^{m_{\bg}+m_{\bg'}}
	A_{\bg} \bar{A}_{\bg'}  
	= \left( \frac{1}{2} \right)^n
	\sum_{j=0}^N (-1)^{2j}
	\left({N\atop j}  \right) 
	= 2^{N-n} \ . 
	\end{equation} 
\end{proof}


We have computed the contributions $C_{\bg}$ to the variance defined in (\ref{contributionC}) for irreducible pseudo orbits with self-intersections.  Substituting in (\ref{avgvarcoeff2}) produces proposition \ref{thm:main}.

\section{Evaluating $C_{\bg}$ for Primitive Pseudo Orbits}
\label{sec:contributions}

In this section we describe a parity argument that shows $C_{\bg}=0$ when $\bg$ is a primitive pseudo orbit where there are $\ell$-encounters of positive length or where $\ell \geq 3$.  This differs from the cancellation scheme used to reduce the primitive pseudo orbit sums to irreducible pseudo orbit sums in \cite{BHJ12}.  As the argument only depends on the sign of the contribution of a partner pseudo orbit to (\ref{contributionC}) and the structure of the pseudo orbit this kind of argument may be able to be extended to other systems.  
 This section employs a connection with Lyndon words, a theme that will recur later, however this section can safely be skipped as later material is not contingent on it.  The first subsection \ref{sec:contributions single enc} describes the main idea in a setting which avoids the need for Lyndon words and can be read as a stand alone argument. 

\subsection{A Single Encounter of Positive Length with Distinct Links}
\label{sec:contributions single enc}

We will start by considering a primitive pseudo orbit with a single encounter of positive length where all the links are distinct.  We introduce the parity argument in this setting first where the proof can be described using permutations.  

Let $\bg$ be a primitive pseudo orbit with one $\ell$-encounter.  The encounter length $\tilde{n} > 0$ and all the links are distinct, see figure \ref{figfigure8enc}(a).
Each primitive partner pseudo orbit $\bg'$ corresponds to an element of $S_{\ell}$, as described in section \ref{sec:single encounter distinct links}.
Let $A_{\link_i}$ denote the product of scattering amplitudes at the interior vertices of $\link_i$ between $v_{\tilde{n}}$ and $v_0$ for $i=1,2, \dots, \ell$.
Similarly, $A_{\mathrm{enc}}$ will denote the product of scattering amplitudes at the interior vertices of the encounter between $v_0$ and $v_{\tilde{n}}$.
We note that $\bg$ will pick up a scattering coefficient at $v_0$ from each of the incoming vertices $s_1, \dots, s_{\ell} $ to the outgoing vertex $v_1$ and a scattering coefficient at $v_{\tilde{n}}$ from the incoming vertex  $v_{\tilde{n}-1}$ to each of the outgoing vertices $f_1, \dots, f_{\ell}$.
Then the stability amplitude of $\bg$ is 
\begin{equation}
A_{\bg} = \left( \prod_{\substack{ h = 1 \cr \gamma_h \textrm{ has no encounters} }}^{m_{\bg}} 	
A_{\gamma_h} \right)
A_{\enc}^\ell
\left( \prod_{i=1}^{\ell} A_{\link_i}  
\VS_{v_1, s_i}^{(v_0)} \VS_{f_i, v_{\tilde{n}-1}}^{(v_{\tilde{n}})} \right) \ .
\end{equation}
Permuting the order in which links are traversed to obtain a partner orbit $\bg'$ does not change the scattering coefficients, so for any partner orbit $A_{\bg'} = A_{\bg}$ and $A_{\bg} \bar{A}_{\bg'} = 2^{-n}$.
So from (\ref{contributionC}),
\begin{equation}
C_{\bg} = \left( \frac{1}{2} \right)^n (-1)^{m_{\bg}}
\sum_{\bg' \in \mathcal{P}_{\bg} }
(-1)^{m_{\bg'}} \ .
\label{contributionwithscattering}
\end{equation}
\begin{lemma}
	For a primitive pseudo orbit $\bg$ with a single $\ell$-encounter of positive length and distinct links, $C_{\bg} = 0$.
	\label{thmCont1}
\end{lemma}
\begin{proof}
	The number of primitive periodic orbits $m_{\bg}$ in $\bg$ is the number of cycles in $\sigma_{\bg}\in S_{\ell}$ plus the number of periodic orbits with no encounters which we denote $u$.
	So,
	\begin{equation}
	(-1)^{m_{\bg}}=\sgn(\sigma_{\bg}) (-1)^{u+\ell} \ .
	\end{equation}
Hence,
	\begin{equation}
	C_{\bg} = \left( \frac{1}{2} \right)^n  
	\sgn(\sigma_{\bg})
	\sum_{\sigma_{\bg'} \in S_\ell}
	\sgn(\sigma_{\bg'}) \ .
	\end{equation}
	As there are equal numbers of odd and even permutations in $S_\ell$, the result follows.
\end{proof} 

The following sections demonstrate that this kind of parity argument extends to primitive pseudo orbits with repeated links and positive length and to primitive pseudo orbits with an $\ell$-encounter of length zero when $\ell\geq 3$ by labeling the pseudo orbits with Lyndon tuples. 

\subsection{A Single $\ell$-encounter with $\ell\geq 2$ of Positive Length}
\label{sec:contributions single lenc positive}

As the encounter length is positive, again $A_{\bg} \bar{A}_{\bg'} = 2^{-n}$.
Now, to allow for repeated links, the set of links defines a multiset $\cM$.  
To each primitive pseudo orbit, there is an associated Lyndon tuple $\tup(w) = (l_1, l_2, \dots, l_k)$ over a word $w$ of length $\ell$ from the set of Lyndon tuples $\mathcal{L}(\cM)$ over $\cM$.
Each Lyndon word $l_i$ in $\tup(w)$ uniquely corresponds to a primitive periodic orbit in $\bg$, so the number of periodic orbits in $\bg'$ is the number of Lyndon words in $\tup(w)$ plus the number of primitive periodic orbits in $\bg'$ with no intersections, which we denote $u$.  
As the Lyndon index $i_{\mathcal{L}}(w) = \ell - k$, we have 
\begin{equation}
m_{\bg'} = \ell - i_{\mathcal{L}}(w) + u \ .
\label{numoforbitsrepeating}
\end{equation}

\begin{lemma}
	For a primitive pseudo orbit $\bg$ containing a single $\ell$-encounter with $\ell\geq 2$ of positive length, $C_{\bg} = 0$.
	\label{thmCont2}
\end{lemma}
\begin{proof}
	Combining (\ref{contributionC}) and (\ref{numoforbitsrepeating}),
	\begin{equation}
	C_{\bg} = \left( \frac{1}{2} \right)^n  (-1)^{i_{\mathcal{L}}(w)}
	\sum_{\tup(w') \in \mathcal{L}(\cM)}
	(-1)^{i_{\mathcal{L}}(w')}  = 0 \ ,
	\end{equation}
	as theorem \ref{thm:splitting} implies there are equal numbers of Lyndon tuples in $\mathcal{L}(\cM)$ with odd and even Lyndon index.
\end{proof}

\subsection{A Single $\ell$-encounter with $\ell\geq 3$ of Zero Length}
\label{sec:contributions single lenc zero}

Now consider a primitive pseudo orbit $\bg$ with a single $\ell$-encounter where $\ell \geq 3$ and the length of the encounter is zero.
We include the possibility of repeated links by using the multiset notation; when the links are distinct each element of the multiset has multiplicity one.
As the self-intersection has zero length, there is a single intersection at vertex $v$ and the product of stability amplitudes depends on how many times the single negative scattering coefficient at $v$, see (\ref{DFTmatrix}), appears in $A_{\bg'}$.
To prove $C_{\bg} =0$, we show every  
Lyndon tuple over $\cM$ with even (odd) Lyndon index can be paired with a Lyndon tuple over $\cM$ with odd (even) Lyndon index where the corresponding pseudo orbits contribute the same number of negative scattering coefficients at $v$.

Consider the single encounter vertex $v$ and the four adjacent vertices of the binary graph, figure \ref{figLencdiagram}.  
We labeled the adjacent vertices preceding the encounter on each repetition in $\bg$ by $s_1, \dots, s_{\ell}$ and those following the encounter by $f_1, \dots, f_{\ell}$ so that $f_j$ was joined to $s_j$ by the $j$-th link.
We denote the two vertices of the binary graph incoming at $v$ by $a$ and $b$ and the outgoing pair at $v$ by $c$ and $d$.   Note that each link must enter $v$ from either $a$ or $b$ and must leave $v$ to either $c$ or $d$.  
In order for $v$ to be traversed by $\bg$ a total of $\ell$ times, $\bg$ must enter through bonds $(a,v)$ and $(b,v)$ a total of $\ell$ times, with each bond used at least once (otherwise the encounter is not of length zero).  
We let $0<I<\ell$ denote the number of times the bond $(b,v)$ is traversed so $(a,v)$ is traversed $\ell-I$ times.
Similarly, we let $0<J<\ell$ denote the number of times $(v,d)$ is used so $(v,c)$ is traversed $\ell-J$ times.
We also assume $J \geq 2$, as $\ell \geq 3$.

\begin{figure}
	\centering
	\includegraphics[scale=1, trim=100 595 325 125, clip]{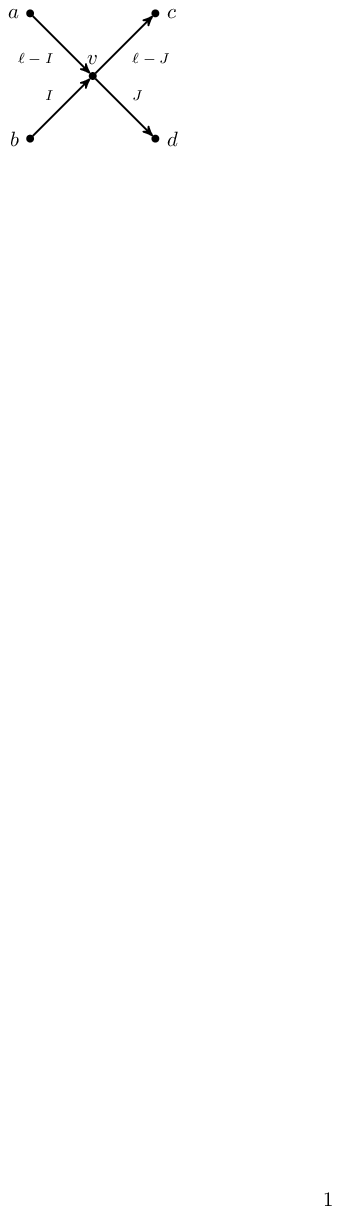}
	\caption[An $\ell$-encounter at vertex $v$ with the incoming and outgoing bonds at $v$]{An $\ell$-encounter at vertex $v$ with the incoming and outgoing bonds at $v$.  Links are not pictured, but each link starts with the bond $(v,c)$ or $(v,d)$ and ends with the bond $(a,v)$ or $(b,v)$.  Each bond is labeled by the number of times it appears in the links of $\bg$.} 
	\label{figLencdiagram}
\end{figure}

Without loss of generality, we assign the single negative scattering coefficient in the scattering matrix (\ref{DFTmatrix}) to $\VS_{d,b}^{(v)}$.
So to compute $C_{\bg}$, we need to track the number of times a link ending with $(b,v)$ is followed by a link beginning with $(v,d)$ in $\bg$ and each of its partner orbits.
Let the multiset $\cD \subset \cM$ be the multiset of links starting at the bond $(v,d)$.  We can define the link labels so that $\cD$ consists of the first $\nu$ labels for $\nu < \mu$, 
\begin{equation}
\cD = [ 1^{m_1}, 2^{m_2}, \dots, \nu^{m_{\nu}} : \link_i \text{ begins with bond } (v,d) \text{ for } 1 \leq i \leq \nu ] \ .
\end{equation}
As $J \geq 2$, so also $|\cD| \geq 2$.  
Then let $\cB \subset \cM$ be the multiset such that
\begin{equation}
\cB = [ i^{m_i} : \link_i \text{ ends with bond } (b,v) ] \ .
\end{equation}
Note that if $\link_i$ is repeated $m_i$ times in $\bg$, then $i$ has multiplicity $m_i$ in $\cM$, and thus has multiplicity $m_i$ in the set $\cB$ and/or $\cD$, if it is present in either.
It is possible for the intersection of multisets $\cB$ and $\cD$ to be non-empty.

Let $T: \mathcal{L}(\cM) \to \mathbb{N}_0$ be a function that counts the number of times a label from the set $\cB$ is followed by a label from the set $\cD$ (cyclically) in the Lyndon words of some Lyndon tuple $\tup(w)$.
So if $\bg$ corresponds to the Lyndon tuple $\tup(w)$, then $T(\tup(w))$ is the number of times that the scattering coefficient $\VS_{d,b}^{(v)}$ appears in $A_{\bg}$.
The following lemma describes the relationship between the bijection $f$ between odd and even Lyndon tuples (\ref{eqPPOmaptuples}), used in the proof of theorem \ref{thm:splitting}, and the function $T$.

\begin{lemma}
	Let $\cM = [ 1^{m_1}, 2^{m_2}, \dots, \mu^{m_{\mu}} ]$ be a multiset where $m_i > 0$ for all $i = 1, \dots, \mu$, and let $\cB, \cD \subset \cM$ be defined as above.
	If $f(w) = w'$, then $T(\tup(w)) = T(\tup(w'))$.
	\label{propL3enczerobalance}
\end{lemma}

\begin{proof}
	Let $\cM$ be a multiset and $\tup(w) = (l_1, l_2, \dots, l_k) \in \mathcal{L}(\cM)$ for $w = l_k \dots l_2 l_1$ which is a strictly decreasing Lyndon decomposition, $l_1 \triangleleft l_2 \triangleleft \cdots \triangleleft l_k$.
	Note that, $l_1$ must begin with the letter one.
	
	If $l_1$ is not splittable, either $l_1$ is a single letter or its standard factorization $(r_1, s_1)$ does not satisfy $s_1 \triangleleft l_2$.
	Assume $l_1 \in \mathcal{A}$; then $l_2$ begins with an element of $\cD$, as $|\cD| \geq 2$.  Thus, the letter $l_1$ is followed (cyclically) by an element of $\cD$, itself, and the last letter of $l_2$ is followed (cyclically) by an element of $\cD$.
	Under the map $f$ the first two words of the Lyndon tuple are combined into $l_0 = l_1 l_2$ and the letter $l_1$ is now followed by an element of $\cD$, the first letter of $l_2$, and the last letter of $l_2$ is now followed by the letter $l_1$, which is also an element of $\cD$.  
	Hence, if $f(w) = w'$, also $T(\tup(w)) = T(\tup(w'))$.

	Suppose $l_1$ has standard factorization $(r_1, s_1)$ with $s_1 \triangleright l_2$, so $l_1$ is not splittable.
	If $l_1$ contains a single element of $\cD$, then as $|\cD| \geq 2$ and $l_1 \triangleleft l_2 \triangleleft \cdots \triangleleft l_k$ with each a Lyndon word, $l_2$ begins with an element of $D$.
	Otherwise, $l_1$ contains multiple elements of $\cD$; as $s_1$ is the minimum proper suffix lexicographically of $l_1$ by proposition \ref{propstdfactsuffix}, it must also begin with an element of $\cD$.
	As $l_2 \triangleleft s_1$ by assumption, $l_2$ must also begin with an element of $\cD$.
	Hence, both the last letter of $l_1$ is followed by an element of $D$ and the last letter of $l_2$ is followed by an element of $\cD$. 
	As $l_1$ is not splittable, the map $f$ combines $l_1$ and $l_2$ into $l_0 = l_1 l_2$.  In $l_0$ the last letter of $l_1$ is still followed by an element of $\cD$, the first letter of $l_2$, and the last letter of $l_2$ is still followed by an element of $\cD$, the first letter of $l_1$. So if $f(w) = w'$ also $T(\tup(w)) = T(\tup(w'))$.

	If $l_1$ is splittable, then $|l_1| \geq 2$ and $l_1$ has standard factorization $(r_1, s_1)$ such that $s_1 \triangleleft l_2$.
	Assume that $l_1$ has only one element of $\cD$, which must be the first letter.  
	Then the least proper suffix $s_1$ of $l_1$ begins with a letter that is greater than all elements of $D$.
	As $|\cD| \geq 2$, $l_2$ must begin with an element of $\cD$, so $s_1 \triangleright l_2$, a contradiction. 
	Then, as $l_1$ must contain at least two elements of $\cD$, consider its standard factorization $(r_1, s_1)$.  As $s_1$ is the least proper suffix of $l_1$ lexicographically by proposition \ref{propstdfactsuffix}, both $r_1$ and $s_1$ begin with elements of $\cD$.  In $l_1$ the last letter of $r_1$ is followed by an element of $\cD$ as is the last element of $s_1$.  
	For a splittable word $l_1$  the map $f(w) = w'$ with $\tup(w') = (r_1, s_1, l_2, \dots, l_k)$.  Then in $\tup(w')$ the last letter of $r_1$ is still followed by an element of $\cD$ and the last letter of $s_1$ is still followed by an element of $\cD$.  So $T(\tup(w)) = T(\tup(w'))$. 
\end{proof}

The partner orbits $\bg' \in \mathcal{P}_{\bg}$ of $\bg$ correspond to the elements of $\mathcal{L}(\cM)$. So lemma \ref{propL3enczerobalance} implies every partner pseudo orbit labeled by an even (odd) $\ell$-word $w_1$, with Lyndon tuple $\tup(w_1)$, can be paired with a partner pseudo orbit labeled with an odd (even) $\ell$-word $w_2$, with Lyndon tuple $\tup(w_2)$, and both elements of the pair use the scattering coefficient $\VS_{d,b}^{(v)}$ the same number of times.
This allows us to evaluate $C_{\bg}$.
\begin{lemma}
	For a primitive pseudo orbit $\bg$ containing a single $\ell$-encounter of length zero with $\ell\geq 3$ we have $C_{\bg} = 0$.
	\label{thmCont3}
\end{lemma}
\begin{proof}
	Let $\bg$ be a primitive pseudo orbit containing a single $\ell$-encounter of length zero with $\ell\geq 3$, where $\bg$ corresponds to a Lyndon tuple $\tup(w) \in \mathcal{L}(\cM)$.
	The partner pseudo orbits $\bg' \in \mathcal{P}_{\bg}$ correspond to all Lyndon tuples $\tup(w') \in \mathcal{L}(\cM)$.  If we group the sum over partner pseudo orbits according to the number of times $\alpha$ that the orbit scatters from $b$ to $d$ at $v$ then using equation (\ref{numoforbitsrepeating}),
	\begin{equation}
	\fl C_{\bg} = \left( \frac{1}{2} \right)^n  (-1)^{T(\tup(w)) + i_{\mathcal{L}}(w)}
	\sum_{\alpha = I-\min\{I,\ell-J\}}^{\min\{I,J\}}
	(-1)^{\alpha}
	\sum_{\substack{\tup(w') \in \mathcal{L}(\cM): \cr T(\tup(w')) = \alpha }}
	(-1)^{i_{\mathcal{L}}(w')}  = 0 \ ,
	\end{equation}
	where lemma \ref{propL3enczerobalance} is used to evaluate the inner sum.
\end{proof}

\subsection{Multiple Encounters}
\label{sec:contributions mult enc various}

The previous lemmas imply the following proposition.
\begin{proposition}
	Let $\bg$ be a primitive pseudo orbit  containing $N$ encounters of types $\vec{\ell} = (\ell_1, \ell_2, \dots, \ell_N)$, such that at least one self-intersection is either of positive length or is an $\ell$-encounter of length zero with $\ell \geq 3$, then $C_{\bg} = 0$.
	\label{thmCont4}
\end{proposition}
\begin{proof}  Following the previous lemmas we only need to establish the result for a primitive pseudo orbit
	$\bg$ with $N\geq 2$ self-intersections  where, 
	without loss of generality, the first encounter has either positive length or it has length zero and $\ell_1 \geq 3$.   We must evaluate the sum over all primitive pseudo orbit partners of $\bg$, equation (\ref{contributionC}).
	
	First consider the subset of primitive partner pseudo orbits $\bg'$ that are obtained by reordering $\bg$ at the first self-intersection.   Starting at each link that is outgoing at the first encounter follow $\bg$ until it returns to the first encounter.  Each of these sequences of links can now be treated as a single link starting and ending at the first encounter.  Then applying lemma \ref{thmCont2} or \ref{thmCont3} we can sum over the subset of primitive partner pseudo orbits of $\bg$ obtained by reordering $\bg$ at the first self-intersection only.  The contribution to $C_{\bg}$ from this subset of partner pseudo orbits is zero.  
	
	Now choose a primitive pseudo orbit partner of $\bg$ which did not appear in the subset generated from $\bg$ by reordering at the first self-intersection, and denote it $\bg_1$.  Starting at the first encounter, follow each outgoing link of $\bg_1$ until it returns to the first encounter and treat each of these sequences of links as a single link starting and ending at the first encounter.  Note that these links must differ from those constructed using $\bg$.  Then applying lemma  \ref{thmCont2} or \ref{thmCont3} we can sum over the subset of primitive partner pseudo orbits of $\bg$ obtained by reordering $\bg_1$ at the first self-intersection only, up to a sign.  However, the contribution to $C_{\bg}$ from this subset of partner orbits is again zero. Repeating the procedure we obtain a sequence of primitive partner pseudo orbits $\bg_2, \bg_3, \dots$, each of which has not been counted previously.  At each step the contribution to $C_{\bg}$ from the subset of primitive pseudo orbit partners generated from $\bg_j$ by reordering links at the first self-intersection is zero.  As there are a finite number of primitive pseudo orbits of length $n$ the sequence terminates with all primitive partner orbits of $\bg$ included in the sum. 
\end{proof}

\section{The Semiclassical Limit}
\label{sec:semiclassical}

For quantum graphs the semiclassical limit is the limit of a sequence of graphs with increasing numbers of bonds, the limit of increasing spectral density \cite{KS99}.  For binary graphs we will consider the family of graphs with $B=p\cdot 2^{r+1}$ and consider the large $r$ limit.  To compare the variance of coefficients of the characteristic polynomial for families of graphs where the size of the graphs, and so also the number of coefficients, varies we can take $r$, or equivalently $n$, to infinity for a fixed ratio $n/B$.  For example, figure \ref{fignumericsBGp1} shows the variance of coefficients of the characteristic polynomial evaluated numerically for the family of binary graphs with $p=1$.  

\begin{proposition}\label{variance limit}
	Consider the family of binary graphs with $V=p\cdot 2^{r}$ vertices with fixed $p$ odd and the ratio $n/B$ fixed. Then,
	\begin{equation}
	\langle |a_n|^2 \rangle_k \approx
	\left( \frac{1}{2} \right)^{n} \cdot \PPO_{V}(n)  \ ,
	\end{equation}
	where $\PPO_{V}(n)$ is the number of primitive pseudo orbits of topological length $n$ on the graph with $V$ vertices.
	\label{prop semiclassical}
\end{proposition}


\begin{proof}
	We seek to asymptotically determine the size of the sets $\hP_{N}^n$ in proposition \ref{thm:main}.
	To have a $2$-encounter of length zero a pseudo orbit must intersect at a single vertex $v$ exactly twice.  Following the pseudo orbit from some arbitrary starting vertex, the second time that the pseudo orbit encounters $v$ it must exit $v$ using the bond that was not used the first time it passed through $v$ in order for the encounter to have length zero. 
	The probability that a pseudo orbit leaves $v$ on a different bond the second time it visits $v$ than it did the first time is $1/2$. Consequently, for long pseudo orbits on large graphs, where the encounters are independent, the probability that all the encounters have length zero is $(1/2)^N$ where $N$ is the number of $2$-encounters.  So for large $n$, 
	\begin{equation}
	| \hP_{N}^n | 
	\approx \left( \frac{1}{2} \right)^N \cdot |\iP_N^n | \ ,
	\label{setapprox}
	\end{equation}
	where $\iP_N^n$ is the number of primitive pseudo orbits with $N$ encounters which are all $2$-encounters but where the encounters can have any length.
	
	Let $(\iP_N^n)^c = \mathcal{P}_N^n \backslash \iP_N^n$ be the set of all primitive pseudo orbits of topological length $n$ with $N$ self-intersections that have $\vec{\ell} = (\ell_1, \ell_2, \dots, \ell_N)$ repetitions, such that at least one $\ell_i \geq 3$.  
	In order to have an $\ell$-encounter with $\ell \geq 3$, after a $2$-encounter the pseudo orbit must return to the first vertex of the encounter sequence $v_0$ for a third time. 
	There are many ways to make an encounter in a pseudo orbit, as there are $n$ possible points of intersection, but only one way to make a $2$-encounter into a $3$-encounter by intersecting a second time at an encounter vertex.
	
	As binary graphs are mixing \cite{T00}, the probability to land on any vertex after a large number of steps is $B^{-1}$.
	Hence, the number of orbits in $(\iP_N^n)^c$ scales with the size of $\mathcal{P}_N^n$ asymptotically like $1/B$.
	Thus, for large $n$, 
	\begin{equation}
	| \iP_N^n | \approx | \mathcal{P}_N^n | \ .
	\label{asympsetsize}
	\end{equation}
	Substituting (\ref{setapprox}) and (\ref{asympsetsize}) in (\ref{variancefinal}),
	\begin{equation}
	\langle |a_n|^2 \rangle_k
	\approx \left( \frac{1}{2} \right)^{n} \sum_{N=0}^n |\mathcal{P}_N^n |  \ .
	\end{equation}
	Noting that the sets $\mathcal{P}_N^n$ are disjoint, their union is the set of all primitive pseudo orbits $\PPO_{V}(n)$ of length $n$.
\end{proof}

The asymptotic, proposition \ref{prop semiclassical}, agrees with the diagonal contribution where every orbit is counted once with equal weight (\ref{diagpart}). 
However, proposition \ref{thm:main} itself, looks quite different, as it depends only on pseudo orbits with no self-intersections and those where all the self-intersections are $2$-encounters of length zero.  
This is somewhat surprising. 
Typically when a dynamical approach is used to evaluate a spectral statistic, a diagonal approximation appears naturally.  However, in previous results an asymptotic argument is required at an early stage.  This is the first example of a spectral statistic, for a chaotic quantum system, that can be evaluated precisely before taking a semiclassical limit. So we see, for the first time, a result that does not appear directly connected to a diagonal approximation, but where this connection, nevertheless, appears asymptotically.

\section{Counting Orbits and Pseudo Orbits}
\label{sec:counting}

Motivated by proposition \ref{prop semiclassical}, we count the number of primitive pseudo orbits on binary graphs with $V=p\cdot 2^r$ vertices where $p$ is odd.  
This was obtained for de Bruijn graphs (the family with $p=1$) in \cite{BHS19}. We start by relating the number of primitive periodic orbits of length $n$ to the number of Lyndon words of length $n$ on a binary alphabet.


\subsection{Counting Primitive Periodic Orbits}
\label{sec:counting prim POs}

On a binary de Bruijn graph with $V=2^r$ vertices there is a bijection between primitive periodic orbits of length $n$ and binary Lyndon words of length $n$.  Consequently, 
\begin{equation}
\PO_V(n) = \Lyndon \ .
\end{equation}
When $V=p\cdot 2^r$ with $p\neq 1$ and odd there is no longer a straightforward relationship between periodic orbits and Lyndon words.  

Let $\A_V$ be the $V\times V$ adjacency matrix of the binary graph with $V=p\cdot 2^r$ vertices. 
Consider $\A_V$ as a block matrix with blocks of dimensions $V/2 \times V/2$, where we label the two blocks in the last row of $\A_V$ by $\B_{V/2}^0$ and $\B_{V/2}^1$.
Thus, we can let $0 \leq i,j < V/2$ for each block matrix $\B^0$ and $\B^1$ and write
\begin{eqnarray}
\left[\B^{0}_{V/2}\right]_{i,j} &= \delta_{2i, j}   
+ \delta_{2i+1, j}  \\
\left[ \B^{1}_{V/2} \right]_{i,j} &= \delta_{2i, j+V/2}   
+ \delta_{2i+1, j+V/2} \ .
\end{eqnarray}
We let $ \A_p = \B_p^0 + \B_p^1$, where $\B_p^0, \B_p^1$ are $p\times p$ blocks in $\A_{2p}$.

\begin{lemma}
	For a binary graph with $V=p\cdot 2^r$ vertices, the nonzero eigenvalues of the adjacency matrix $\A_V$ are also eigenvalues of $\A_p$ with the same multiplicity; the multiplicity of zero as an eigenvalue of $\A_V$ is at least $V-p$.
	\label{thmnumnonzeroevals}
\end{lemma}
\begin{proof}
	Consider the matrix $\A_V-\lambda \UI_V$, for $V=p\cdot 2^r$ with $r \geq 1$, as a block matrix with blocks of size $V/2 \times V/2$.  For a binary graph, $\A_V$ was defined in (\ref{eq:defn A_V}) and
	subtracting the second row of blocks of $\A_V-\lambda \UI_V$ from the first row of blocks eliminates the entries of one in the top row of blocks leaving entries of $\pm \lambda$. 
	Hence,
	\begin{equation}
	\det(\A_V - \lambda \UI_V) = 
	\left| 
	\begin{array}{cc}
	-\lambda \UI_{V/2} & \lambda \UI_{V/2}\\
	\B_{V/2}^0 & \B_{V/2}^1- \lambda \UI_{V/2} \\
	\end{array}
	\right| \ .
	\end{equation}
	Adding the first column of blocks to the second makes the block matrix triangular and
	\begin{eqnarray}
	[\B_{V/2}^0 + \B_{V/2}^1]_{i,j}
	&= \delta_{2i, j} + \delta_{2i+1, j} + \delta_{2i, j+V/2} + \delta_{2i+1, j+V/2} \\
	&= \delta_{2i \pmod{V/2}, j} + \delta_{2i+1 \pmod{V/2}, j} \\
	&= \left( \A_{V/2} \right)_{i,j} \ .
	\end{eqnarray} 
	Consequently,
	\begin{equation}
	\det(\A_V - \lambda \UI_V)
	= (-\lambda)^{V/2} \det(\A_{V/2} - \lambda \UI_{V/2})\ .
	\end{equation}
	So the only non-zero eigenvalues of $\A_V$ are those of $\A_{V/2}$.   Iterating, the non-zero eigenvalues of $\A_V$ for $V=p\cdot 2^r$ are also eigenvalues of $\A_p$.
\end{proof}

As was the case for $p=1$, lemma \ref{thmnumnonzeroevals} implies the number of primitive periodic orbits when $p\neq 1$ remains independent of $r$.  So we write $\PO_{p}(n)$ for $\PO_V(n)$ henceforth.

We now seek to determine the eigenvalues of $\A_p$.
We will write $\A_p$ in a form that allows us to apply results for eigenvalues of a generalized permutation matrix, a matrix with exactly one non-zero entry in each row and column.  
These matrices factor as a product of a diagonal matrix and a permutation matrix \cite{GM15}.
Let $\pi$ be the permutation associated to a generalized permutation matrix and let $(i_1\ i_2\ \dots\ i_c)$ be a $c$-cycle of $\pi$ for $1 \leq c \leq p$. Let the nonzero entries of the generalized permutation matrix in the  columns $i_1, i_2, \dots, i_c$ be denoted $a_{i_1}, a_{i_2}, \dots, a_{i_c}$.
Then a factor of the characteristic polynomial of the generalized permutation matrix is \cite{GM15},
\begin{equation}
\left( \lambda^c - \prod_{j=1}^{c} a_{i_j} \right) \ .
\label{charpolyfactor}
\end{equation} 

Let $\S$ be the $p\times p$ permutation matrix
\begin{equation}
\S = \left( 
\begin{array}{ c c c c c  }
0 & 1 & 0 &  \dots & 0 \\
0 & 0 & 1 &\dots & 0 \\
\vdots & \vdots & \ddots  &\ddots  & \vdots  \\	
0 & 0 & \dots & 0  & 1 \\
1 & 0 & \dots & 0  & 0 \\ 
\end{array}
\right) \ ,
\end{equation}
and note that $\S^p = \UI_p$.  
As $\S$ is a circulant matrix, it can be diagonalized by the Discrete Fourier Transform matrix \cite{D79}, which we use in the proof of the next lemma.
We also define the $p\times p$ matrix $\H=\diag \{ 1,0,\dots,0\}$.
We can now write,
\begin{equation}\label{eq:Ap decomposition}
\A_p = \left( \sum_{h=0}^{p-1} \S^{p-h} \H \S^{2h} \right) (\UI_p+\S) \ .
\end{equation}
Multiplying  $\H$ on the left by $\S^{p-h}$ moves the unit entry in row zero of $\H$ to row $h$ and subsequent right-multiplication by $\S^{2h}$ sends the  non-zero entry in the first column to the column $2h \pmod{p}$.
As the original adjacency matrix $\A_V$ had two consecutive entries in each row, right multiplication by the matrix $\UI_p+\S$ duplicates the first non-zero entry in each row in the next consecutive column in $\A_p$.
Note that, $\A_1 = [ 2 ]$.


\begin{lemma}
	\label{thmevalsofAp}
	The matrix $\A_p$ has a simple eigenvalue of $2$ and all other eigenvalues form complete sets of roots of unity.
\end{lemma}
\begin{proof}
	Let $\F$ be a $p\times p$ Discrete Fourier Transform matrix. Then, 
	\begin{equation}
	\F^{-1} \S \F = \D = \diag(1, \xi, \xi^2, \dots, \xi^{p-1}) \ ,
	\end{equation}
	where $\xi = \rme^{2\pi\rmi /p}$, the primitive $p$-th root of unity.
	Define,
	\begin{eqnarray}
	\tilde{\A}_p	&=   \F^{-1} \A_p \F \\
	&=\F^{-1} \left( \sum_{h=0}^{p-1} 
	\S^{p-h} \H \S^{2h} + \S^{p-h} \H \S^{2h+1} \right) \F \\
	&= \frac{1}{p} 
	\left( \sum_{h=0}^{p-1}  
	\D^{p-h} \1_p \D^{2h}
	+ \D^{p-h} \1_p \D^{2h+1}\right)
	\end{eqnarray}
	where we used $\F^{-1} \H \F = \frac{1}{p} \cdot \1_p$ and $\1_p$ is a $p\times p$ matrix where all the entries are $1$.
	%
	%
	The $i,j$-th entry of $\tilde{\A}_p$ for $0 \leq i,j \leq p-1$ is
	\begin{eqnarray}
	[\tilde{\A}_p]_{i,j} 
	&= \frac{1}{p} (1+\xi^j)
	\sum\limits_{h=0}^{p-1} \xi^{h(2j-i)} \\
	&= \cases{
	1+\xi^j & if  $2j-i \pmod{p} \equiv 0$ \\
	0 & otherwise } \ , 
	\label{genpermmatrix}
	\end{eqnarray}
	since $\xi$ is the primitive $p$-th root of unity.
	
	To see that $\tilde{\A}_p$ is a generalized permutation matrix, consider two columns of $\tilde{\A}_p$ numbered $j_1, j_2 \in \{0,1,\dots,p-1\}$, with $j_1 \neq j_2$.
	Then $2j_1 \not\equiv 2j_2 \pmod{p}$.  
	Consider the two rows $i_1, i_2 \in \{0,1,\dots,p-1\}$.  If $2j_1 - i_1 \pmod{p} \equiv 0$ and $2j_2 - i_2 \pmod{p} \equiv 0$ then $i_1 \neq i_2$.
	So each column has a single nonzero entry in a row that is different from the row containing the nonzero entry of any other column.
	Similarly each row has a single nonzero element in a column that differs from the column containing the nonzero entry of any other row. So $\tilde{\A}_p$ is a generalized permutation matrix.
	
	To determine the characteristic polynomial of $\tilde{\A}_p$, we must know the cycle structure of the permutation $\pi\in S_p$ associated to $\tilde{\A}_p$.
	To determine which entries of $\tilde{\A}_p$ are contained in a $c$-cycle of $\pi$ for some $1 \leq c \leq p$, consider powers of $\tilde{\A}_p$.  
	In $\tilde{\A}_p$, the only non-zero entry of column $j$ occurs in row $2j \pmod{p}$. 
	If we multiply $\tilde{\A}_p$ by itself, the non-zero entry in column $j$ of $(\tilde{\A}_p)^2$ will be determined by the position of the non-zero entry in column $2j \pmod{p}$; the only non-zero entry in this column is in row $2^2 j \pmod{p}$.
	It will take $c$ iterations of $\tilde{\A}_p$ for the non-zero entry in column $j$ to return to row $2j \pmod{p}$.  
	Consequently, the condition that a $c$-cycle belongs to the disjoint cycle decomposition of the permutation $\pi$ associated to $\tilde{\A}_p$ is
	\begin{equation}
	\label{cyclecondition}
	2^{c} j \equiv j \pmod{p} \ ,
	\end{equation}
	for some column number $j$ and some positive integer $c$.
	
	If we multiply the nonzero entries as defined in (\ref{genpermmatrix}) in the columns numbered 
	\begin{equation}
	j, \ 2j \pmod{p}, \ 2^2j \pmod{p}, \dots, \ 2^{c-1}j \pmod{p}
	\end{equation}
	we obtain a divisor of the characteristic polynomial of $\tilde{\A}_p$ (\ref{charpolyfactor}), 
	\begin{equation}
	\label{qcharpolyfactor}
	\left[ \lambda^c - (1+\xi^j)(1+\xi^{2j}) \cdots (1+\xi^{2^{c-1}j}) \right] \ .
	\end{equation}
	Every generalized permutation will contain a $1$-cycle, as $2j \equiv j \pmod{p}$ for $j=0$, and the nonzero entry in the $j=0$ column is $1 + \xi^0 = 2$.
	This proves the first part of the lemma, as the characteristic polynomial of $\tilde{\A}_p$ always has a divisor $(\lambda-2)$.
	
	It remains to show that any other cycle in the permutation associated to $\tilde{\A}_p$ produces a divisor $(\lambda^c - 1)$ of the characteristic polynomial, producing eigenvalues that are complete sets of $c$-th roots of unity. 
	We first show that, 
	\begin{equation}
	\label{qproductcycle}
	(1+\xi^j)(1+\xi^{2j}) \cdots (1+\xi^{2^{c-1}j}) 
	= \sum_{h=0}^{2^c - 1} \xi^{jh} \ .  
	\end{equation}
	Note that if $c=1$, then this statement clearly holds, as the left-hand side contains only the first factor. 
	Assume (\ref{qproductcycle}) holds for $c-1$ with $c\geq 2$.
	Multiplying by $(1+\xi^{2^{c-1}j})$ the right-hand side then reads,
	\begin{equation}
	\sum_{h=0}^{2^{c-1}-1} 
	\xi^{jh} + \xi^{j(2^{c-1}+h)}
	= \sum_{h=0}^{2^c - 1} \xi^{jh} \ ,
	\end{equation}
	which establishes (\ref{qproductcycle}).
	
	It follows from (\ref{cyclecondition}) that $(2^c - 1)j \equiv 0 \pmod{p}$ for all $j$ in a given cycle of length $c$.
	As $\xi$ is the primitive $p$-th root of unity, and since $p$ consecutive powers of a $p$-th root of unity sum to zero, for $j \neq 0$,
	\begin{equation}
	\sum_{h=1}^{2^c-1} (\xi^j)^h = 0 \ . 
	\end{equation}
	Thus from (\ref{qcharpolyfactor}), the characteristic polynomial has divisors $(\lambda^c - 1)$ for each cycle of length $c \geq 2$ of the permutation $\pi$ associated to the generalized permutation matrix $\tilde{\A}_p$.
	For fixed points of $\pi$, the characteristic polynomial has a divisor $(\lambda-1)$, except in the case of the cycle of length one which corresponds to the first column of $\tilde{\A}_p$ and yields a divisor $(\lambda-2)$.
	\end{proof}

Combining lemmas \ref{thmnumnonzeroevals} and \ref{thmevalsofAp}, the matrix $\A_p$ has no eigenvalues of zero and therefore, zero is an eigenvalue of $\A_V$ with multiplicity $V-p$.
Note, it is clear that $\A_V$ defined in (\ref{eq:defn A_V}) must have an eigenvalue of $2$ with a constant vector as the corresponding eigenvector.
The lemmas agree with the previous results \cite{T01} for $p=1$, as $\A_1 = \left[ 2 \right] $, implying that all other eigenvalues of $\A_V$ are zero when $p=1$.
Thus, ordering the eigenvalues $|\lambda_0| \geq |\lambda_1| \geq |\lambda_2| \geq \cdots |\lambda_{V-1}|$ of $\A_V$, the spectral gap of $\A_V$ is
\begin{equation}
|\lambda_0 |- |\lambda_1| = \cases{ 
2, & if $p=1$ \\
1& otherwise\\} \ ; 
\end{equation}
as when $j\neq 0$ we have $|\lambda_j| =1 $. 

\begin{corollary}
	For a binary graph with $V=p\cdot 2^r$ vertices, let $\beta_{p}(c)$ count the number of cycles of length $c>1$ associated to $\tilde{\A}_p$.  
	Let $\beta_{p}(1) + 1$ be the number of $1$-cycles associated to $\tilde{\A}_p$ so that $\beta_{p}(1)$ counts the number of $1$-cycles, excluding the $1$-cycle that yields the eigenvalue $2$.
	Then
	\begin{equation}
	\label{traceAnevals}
	\Tr((\A_V)^n) = 2^n + \sum_{d|n} d \cdot \beta_{p}(d) \ .
	\end{equation}
	\label{thmtraceAn}
\end{corollary}
\begin{proof}
	The trace of the $n$-th power of $\A_V$ is the sum of the $n$-th powers of the eigenvalues of $\A_V$.  
	The sum of the $n$-th powers of a complete set of $d$-th roots of unity is $d$ if $d|n$ and zero otherwise.  
	Thus, if $\beta_{p}(d) \neq 0$ and $d|n$, there are $\beta_{p}(d)$ complete sets of $d$-th roots of unity, and we get a contribution of $d \cdot \beta_{p}(d)$ to the sum.
\end{proof}

Notice that, the trace of $(\A_V)^n$ counts the number of closed paths of length $n$ on the graph, both primitive and non-primitive. 
In particular, 
\begin{equation}
\label{traceAnpaths}
\Tr((\A_V)^n) = \sum_{d|n} d\cdot \PO_p(d)
= \underbrace{\sum_{\substack{d|n \cr d \neq n}} d\cdot \PO_{p}(d)}_{\substack{\textrm{number of non-primitive} \cr \textrm{closed paths}}}
+ \underbrace{n \cdot \PO_{p}(n)}_{\substack{\textrm{number of primitive} \cr \textrm{closed paths}}}  \ , 
\end{equation}
where the factors of $d$ and $n$ account for each of the distinct vertices that a closed path could begin on within each periodic orbit equivalence class.

To extend the result in \cite{BHS19} for binary de Bruijn graphs, we 
use the following classical lemma \cite{L83}.

\begin{lemma}
	\label{thmclassicLW}
	\begin{equation}
	\sum_{d|n} d \cdot \text{L}_q(d) = q^n 
	\end{equation}
\end{lemma}
The following proposition relates the number of primitive periodic orbits of length $n$ to the number of Lyndon words of length $n$ and the number of cycles of length $n$ associated to $\tilde{\A}_p$.
\begin{proposition}
	\label{thmPO-LW}
	For a binary graph with $V = p\cdot 2^r$ vertices, the number of primitive periodic orbits of length $n$ is $\PO_{p}(n) = \Lyndon + \beta_{p}(n)$.
	Hence, for $n\geq p$, we have $\PO_{p}(n) = \Lyndon $. 
\end{proposition}
\begin{proof}
	First, note that for $\beta_{p}(1) \geq 0$, 
	\begin{equation}
	\PO_{p}(1) 
	= \Tr(\A_V) 
	= 2 + \beta_{p}(1) 
	= \text{L}_2(1) + \beta_{p}(1). 
	\label{orbitslengthone}
	\end{equation}
	
	Let $n$ be a prime.  Then by corollary \ref{thmtraceAn}, lemma \ref{thmclassicLW}, and using (\ref{traceAnpaths}) and (\ref{orbitslengthone}),
	\begin{eqnarray}
	\PO_{p}(1) + n\PO_{p}(n) 
	&= 2^n + \beta_{p}(1) + n \beta_{p}(n) \\
	&= \text{L}_2(1) + n\Lyndon + \beta_{p}(1) + n \beta_{p}(n) \ ,
	\end{eqnarray}
	so $\PO_{p}(n) = \Lyndon + \beta_p(n)$.
	
	Now let $n$ be a product of $m\geq 2$ primes and assume that the result holds for all divisors $d$ of $n$ that are a product of less than $m$ primes.  
	Then,
	\begin{equation}
	\sum_{d|n} d \cdot \PO_{p}(d) 
	= \Tr((\A_V)^n) 
	= \sum_{d|n} d \cdot \left[ \text{L}_2(d) + \beta_{p}(d) \right] \ ,
	\end{equation}
	and the result follows by induction on $m$.  
	There are no cycles of length greater than $p$ associated to a $p\times p$ generalized permutation matrix and $\tilde{\A}_p$ always has one cycle of length one, so cycles associated to $\tilde{\A}_p$ are at most of length $p-1$.  Hence $\PO_{p}(n) = \Lyndon$ for $n \geq p$.
\end{proof}

When $p=1$ we have $\A_1 = [2]$ and $\Tr((\A_{2^r})^n) = 2^n$, so 
\begin{equation}
\sum_{d|n} d \cdot \PO_{1}(d) = \sum_{d|n} d \cdot \text{L}_2(d) \ ,
\end{equation}
and by induction on the number of divisors of $n$, $\PO_1(n) = \Lyndon$ for all $n$, which is the result used in \cite{BHS19}.

\subsection{Counting Primitive Pseudo Orbits}
\label{sec:counting prim PPOs}

Now that we can count primitive periodic orbits, we can also count primitive pseudo orbits.  
When $p=1$, we have a de Bruijn graph, and counting the numbers of primitive pseudo orbits of a particular length is equivalent to counting the number of strictly decreasing Lyndon decompositions of words of the same length, which was accomplished with a generating function argument in \cite{BHS19}.
As we will use a similar argument when $p > 1$, we first review the $p=1$ result.

\subsubsection{Primitive Pseudo Orbits on Binary de Bruijn Graphs}
\label{sec:PPOs de Bruijn}
For de Bruijn graphs, the Lyndon decomposition (see theorem \ref{thmCFL}) of a binary word of length $n$ corresponds to a pseudo orbit of length $n$ on a binary graph.  Each Lyndon word in the decomposition corresponds uniquely to a primitive periodic orbit.  
As a primitive pseudo orbit does not contain any repeated periodic orbits, primitive pseudo orbits correspond uniquely to words with strictly decreasing Lyndon decompositions, as then no Lyndon words are repeated in the decomposition.

For example, the Lyndon decomposition of the binary words of length four are shown below, with parentheses indicating the decomposition.
Words with strictly decreasing Lyndon decomposition are in bold font.

\begin{center}
	\begin{tabular}{ c  c  c  c }
		(0)(0)(0)(0) & \textbf{(0001)} & \textbf{(001)(0)} & \textbf{(0011)} \\
		(01)(0)(0) & (01)(01) & \textbf{(011)(0)} & \textbf{(0111) }\\
		(1)(0)(0)(0) & \textbf{(1)(001)} & \textbf{(1)(01)(0)} & \textbf{(1)(011)} \\
		(1)(1)(0)(0) & (1)(1)(01) & (1)(1)(1)(0) & (1)(1)(1)(1) 
	\end{tabular}
\end{center}

In \cite{BHS19}, the number of strictly decreasing Lyndon decompositions is evaluated for words over $q$-nary alphabets; here we focus on the binary case.
The number of strictly decreasing decompositions of binary words of length $n$ where $n\geq 2$ is
\begin{equation}
\textrm{Str}_2(n) = 2^{n-1} \ .
\label{numstrdec}
\end{equation}
In the previous example, for binary words of length four there are $\textrm{Str}_2(4) = 8$ words with strictly decreasing Lyndon decompositions.

Define the generating function,
\begin{equation}
P(x) = \sum_{n=0}^{\infty} \textrm{Str}_2(n) \cdot x^n \ ,
\end{equation}
where $\textrm{Str}_2(0)= 1$ and $\textrm{Str}_2(1) = 2$.
This generating function is equivalently, 
\begin{equation}
P(x) = \prod_{l=1}^{\infty} (1+x^l)^{\text{L}_2(l)} \ ,
\end{equation}
as the set of all words with strictly decreasing Lyndon decompositions is in bijection with the set of all subsets of Lyndon words (without repetition).
It is clear that any collection of unique Lyndon words can be ordered lexicographically and will correspond to a word with strictly decreasing decomposition.  
The Chen-Fox-Lyndon theorem shows that any word with a strictly decreasing decomposition will correspond uniquely to a subset of Lyndon words without repetition.

To obtain (\ref{numstrdec}), it is sufficient to show $P=F$ on some interval, where 
\begin{equation}\label{eq:defn F}
F(x) = \frac{2x^2-1}{2x-1} = 1 + 2x + \sum_{n=2}^{\infty} 2^{n-1} x^n \ .
\end{equation}
This is accomplished in \cite{BHS19} by noting that $P(0) = F(0) = 1$ and showing,
\begin{equation}
\frac{\rmd}{\rmd x} \log P = \frac{\rmd}{\rmd x} \log F
\end{equation}
on the interval $(-1,1)$ and consequently $\PPO_1(n)=2^{n-1}$ for $n\geq 2$.

\subsubsection{Primitive Pseudo Orbits on General Binary Graphs}
\label{sec:PPOs general}

When $V = p\cdot 2^r$ with $p>1$ and $n$ is sufficiently large, we show that the number of primitive pseudo orbits on the graph, $\PPO_{p}(n)$, is a constant multiple of the $p=1$ result for de Bruijn graphs.  The constant depends on the cycle structure of the permutation associated to $\tilde{\A}_p$.

We define the generating function for the number of primitive pseudo orbits of length $n$,
\begin{equation}
P(x) = \sum_{n=0}^\infty \PPO_{p}(n) \cdot x^n \ ,
\end{equation}
with $\PPO_{p}(0) = 1$ and $\PPO_{p}(1) = 2$.
Which is equivalently,
\begin{equation}
P(x) = \prod_{l=1}^\infty (1+x^l)^{\PO_{p}(l)} \ .
\end{equation}

According to proposition \ref{thmPO-LW}, for a given $p$ there are $\beta_{p}(n)$ additional primitive periodic orbits of length $n$ beyond the number of Lyndon words.  The number $\beta_{p}(n)$ comes from the number of $n$-cycles in the cycle decomposition of the permutation associated to $\tilde{\A}_p$.
Therefore,
\begin{equation}
P(x) = (1+x^{c_1}) \cdots (1+x^{c_{\alpha}}) 
\prod_{l=1}^\infty (1+x^l)^{\text{L}_2(l)} \ , 
\end{equation}
where each $c_j$ is the length of a cycle of the generalized permutation $\tilde{\A}_p$ (excluding the cycle of length one corresponding to $\lambda_0 = 2$), $\alpha+1$ is the number of cycles, and
\begin{equation} 
\label{sumofcycles}
\sum_{j=1}^{\alpha} c_j = p-1 \ .
\end{equation}
Note that, these $c_j$ need not differ from one another.  We now prove proposition \ref{thmNumPPOs} in the introduction, although we can now state it more precisely including the formula for $C_p$.
\begin{proposition}
	\label{thmNumPPOs final}
	For a binary graph with $V=p\cdot 2^r$ vertices, the number of primitive pseudo orbits of length $n > p$ is 
	\begin{equation}
	\PPO_{p}(n) = C_{p} \cdot 2^{n-1} \ ,
	\label{eq:PPOpq 2}
	\end{equation}
	where
	\begin{equation}
	C_{p} = \cases{
	1 & when  $p=1$, \\
	\prod\limits_{j=1}^{\alpha} (1+2^{-c_j}) & when  $p > 1$,\\} 
	\label{eq:Cpq}
	\end{equation}
	and the $c_j$ for $j=1,\dots ,\alpha$ are the lengths of the cycles of $\tilde{\A}_p$ excluding the cycle of length one corresponding to $\lambda_0 = 2$. 
	Furthermore, $C_{p}$ is bounded above by a constant that grows at most linearly in $p$,
	\begin{equation}
	\label{eq:Cpqbound}
	1 \leq C_{p} \leq \frac{3}{2} (p-1) \ . 
	\end{equation}
\end{proposition}

\begin{proof}
	Let $F$ be the function defined in (\ref{eq:defn F}).  We note that 
	\begin{equation}
	\fl (1+x^c) F(x) 
	= 1 + 2x + \cdots + 
	[2^{c-1} + 1] x^c + 
	[2^c + 2] x^{c+1} 
	+ (1+2^{-c}) \sum_{n=c+2} 2^{n-1} x^n \ .
	\end{equation}
	If there is a single cycle associated to $\tilde{\A}_p$ its length is $p-1$.  
	So for $n > p$, the constant $C_{p} = (1 + 2^{-(p-1)})$ and the number of primitive pseudo orbits of length $n$ is $\PPO_{p}(n) = (1 + 2^{-(p-1)}) 2^{n-1}$.
	
	Now suppose the cycle lengths associated to $\tilde{\A}_p$ are $c_1, c_2, \dots, c_\alpha$. Using (\ref{sumofcycles}),
	\begin{equation}
	2 + \sum_{j=1}^{\alpha-1} c_j = p - c_{\alpha} + 1 \ .
	\end{equation}
	We assume the inductive hypothesis,
	\begin{equation} 
	\label{IHconstant}
	\prod_{j=1}^{\alpha-1} (1+x^{c_j}) F(x)	
	= \sum_{n=0}^{p-c_{\alpha}} r_n x^n 
	+ \prod_{j=1}^{\alpha-1} (1 + 2^{-c_j}) 
	\left[ \sum_{n=p-c_{\alpha} +1}^\infty 2^{n-1} x^{n} \right] \ ,
	\end{equation}
	where $r_n$ is the coefficient of $x_n$ for $n=0, 1, \dots, p-c_{\alpha}$.
	Multiplying both sides by $(1+x^{c_{\alpha}})$, we find that the powers of $x$ that are at least $p+1$ are given by
	\begin{equation}
	\prod_{j=1}^{\alpha-1} (1 + 2^{-c_j})
	\left[ \sum_{n=p+1}^\infty 
	2^{n-1} (1+2^{-c_{\alpha}}) x^{n}   \right] \ , 
	\end{equation}
	which is the result in (\ref{eq:Cpq}).
	
	The lower bound on $C_p$ is obtained when $p=1$.
	As the cycle lengths $c_j$ sum to $p-1$ (\ref{sumofcycles}), the greatest number of cycles that could be associated to $\tilde{\A}_p$ would be $p$ when $c_j=1$ for all $j$, yielding the upper bound for $C_{p}$.
\end{proof}

Having counted the number of primitive pseudo orbits on binary graphs in proposition \ref{thmNumPPOs final}, an immediate consequence is that we now know corollary \ref{cor:asymp} for the large graph asymptotic of the variance which was evaluated in terms of the number of primitive pseudo orbits on the graph in proposition \ref{variance limit}.
For the same reason, the diagonal contribution to the variance of the coefficients, defined in (\ref{diagpart}), is 
$\langle |a_n|^2 \rangle_{\diag}=C_p/2$ where $p<n<B-p$.  

 \section{Examples}
\label{sec:examples}

We illustrate the results by computing the variance of the coefficients for binary graphs numerically. 
To do this we generate the characteristic polynomial  (\ref{charpoly matrix}) of $\U (k)=\BS \rme^{\rmi k\ML}$ with the $B$ bond lengths uniformly distributed random numbers in the interval $[0.9, 1.1]$.  The variance of the coefficients of the characteristic polynomial are averaged over an interval of the $k$-spectrum consisting of approximately $50$ million mean spacings. 
The mean spacing of the square roots of the graph's eigenvalues $k$ is $\pi$ divided by the total length of the graph according to the Weyl law, see section \ref{sec:binaryQG}.  For the large values of $B$, MatLab was run on Baylor's Kodiak high performance computing cluster.
However, for $B=160$, the simulation appeared to have converged quickly, so we ended the simulation at $26.5$ million mean spacings.   For $B=192$, the simulation was ended manually at $20.9$ million mean spacings and for $B=320$, the simulation met the maximum run time on Kodiak after $23.9$ million mean spacings.

First we consider proposition \ref{thm:main} where the variance is given precisely, in terms of the number of pseudo orbits of a certain length with no self-intersections or where all the self-intersections are $2$-encounters of length zero.  

\subsection{The Binary de Bruijn Graph with $8$ Vertices}
\label{sec:example V=8}

	The binary de Bruijn graph with $8$ vertices and $16$ bonds is shown in figure \ref{figBGp1-2}.
As $p=1$, the number of primitive periodic orbits of topological length $n$ is the number of binary Lyndon words, $\PO_1(n) = \text{L}_2(n)$.
Moreover, the number of primitive pseudo orbits of topological length $n$ is $\PPO_1(n) = 2^{n-1}$.
These primitive pseudo orbits can be sorted to find those without self-intersections, and those where all the self-intersections are $2$-encounters of length zero.  
Table \ref{figExTableVariance} shows the sizes of these sets for all applicable numbers of self-intersections for $0 \leq n \leq 8$.
For a description of how these set sizes are determined using Lyndon words, see Appendix A.
Proposition \ref{thm:main} is applied to obtain the variance.  
This can be compared to the numerically computed variance $\langle |a_n|^2 \rangle_k$.  
The error between proposition \ref{thm:main} and the numerically generated variance is shown in the final column.  In each case the numerical result shows agreement to at least four decimal places. 

	\begin{figure}[htb]
	\centering
	\includegraphics[scale=1, trim=145 550 120 110, clip]{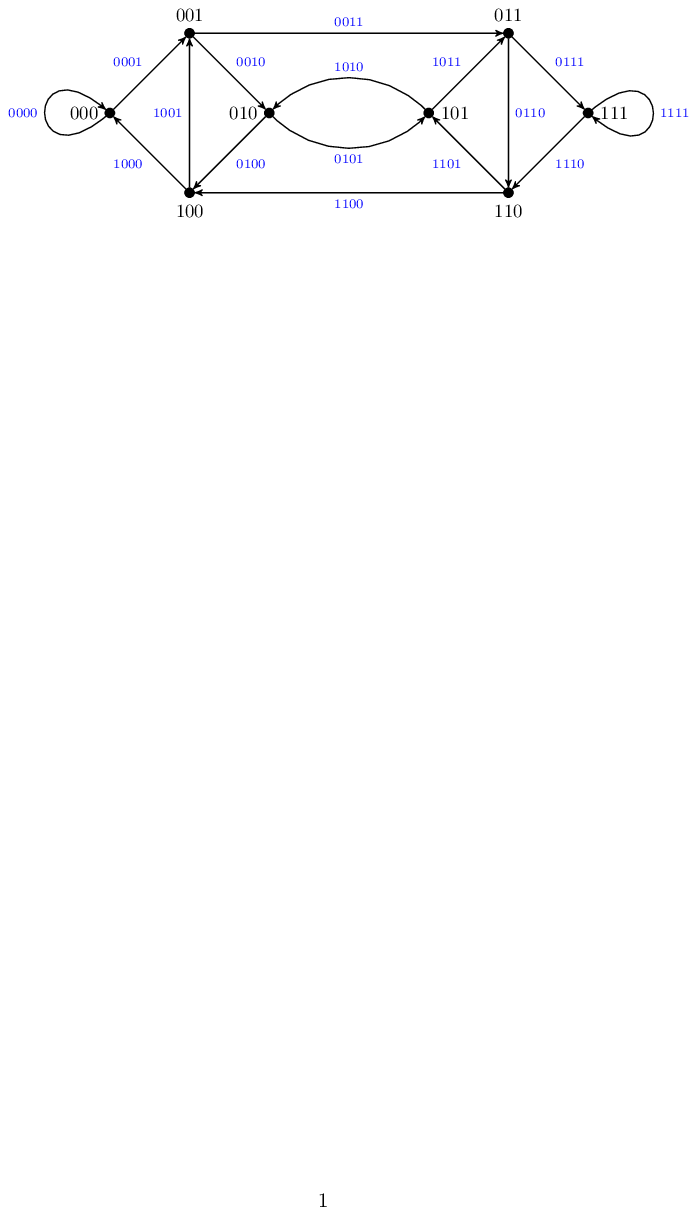}
	\caption{The binary graph with $V=2^3$ vertices and $B=2^4$ bonds}
	\label{figBGp1-2}
\end{figure}

	\begin{table}[htb]
	\caption{\label{figExTableVariance}For a binary graph with $8$ vertices the table shows the sizes of the sets of primitive pseudo orbits with no self-intersections, with one $2$-encounter of length zero and with two $2$-encounters of length zero along with the resulting variance for the first half of the coefficients of the characteristic polynomial. The last two columns provide numerical values of the variance for a set of uniformly distributed random bond lengths averaged over an interval of $50$ million mean spacings and the error between proposition \ref{thm:main} and the numerics.}
	\begin{indented}
		\item[]
	\begin{tabular}{ c c c c c c c }
		$n$ & $|\mathcal{P}_0^n|$ & $|\widehat{\mathcal{P}}_{1}^n|$ 
		& $|\widehat{\mathcal{P}}_{2}^n|$ & $\langle |a_n|^2 \rangle_k$ 
		& Numerics & Error \\
		\hline
		0 & 1 & 0 & 0 & 1 & 1.000000 & 0.000000 \\
		1 & 2 & 0 & 0 & 1 & 0.999991 & 0.000009 \\
		2 & 2 & 0 & 0 & 1/2 & 0.499999 & 0.000001 \\
		3 & 4 & 0 & 0 & 1/2 & 0.499999 & 0.000001 \\
		4 & 8 & 0 & 0 & 1/2 & 0.499999 & 0.000001 \\
		5 & 8 & 8 & 0 & 3/4 & 0.749998 & 0.000002 \\
		6 & 8 & 20 & 0 & 3/4 & 0.749986 & 0.000014 \\
		7 & 16 & 16 & 8 & 5/8 & 0.624989 & 0.000011 \\
		8 & 16 & 16 & 24 & 9/16 & 0.562501 & -0.000001 \\
		\hline
	\end{tabular}
   \end{indented}
\end{table}

\subsection{The Binary Graph with $6$ Vertices}
\label{sec:example V=6}

The binary graph with $V=6$ vertices and $B=12$ bonds is shown in figure \ref{figBGp3}.
As $p=3$, the number of primitive periodic orbits of length $n$ is exactly the number of binary Lyndon words, $\PO_{3}(n) = \text{L}_2(n)$ for $n\neq 2$, and $\PO_{3}(2)=\text{L}_2(2)+1$.
Moreover, the number of primitive pseudo orbits of topological length $n$ is $\PPO_{3}(n) = 5\cdot 2^{n-3}$ for $n > 3$.
These primitive pseudo orbits can be sorted to find those without self-intersections, and where all the self-intersections are $2$-encounters of length zero.
The sizes of these sets for relevant numbers of self-intersections are shown in   
table \ref{figExTableVariance} for $0 \leq n \leq 6$.
For a description of how these set sizes are determined using Lyndon words, see Appendix B.
Proposition \ref{thm:main} is used to obtain the variance in the fourth column.  This is to be compared with the numerically generated  variance,  with the error shown in the last column.

\begin{figure}[ht]
	\centering
	\includegraphics[scale=1, trim=145 550 250 125, clip]{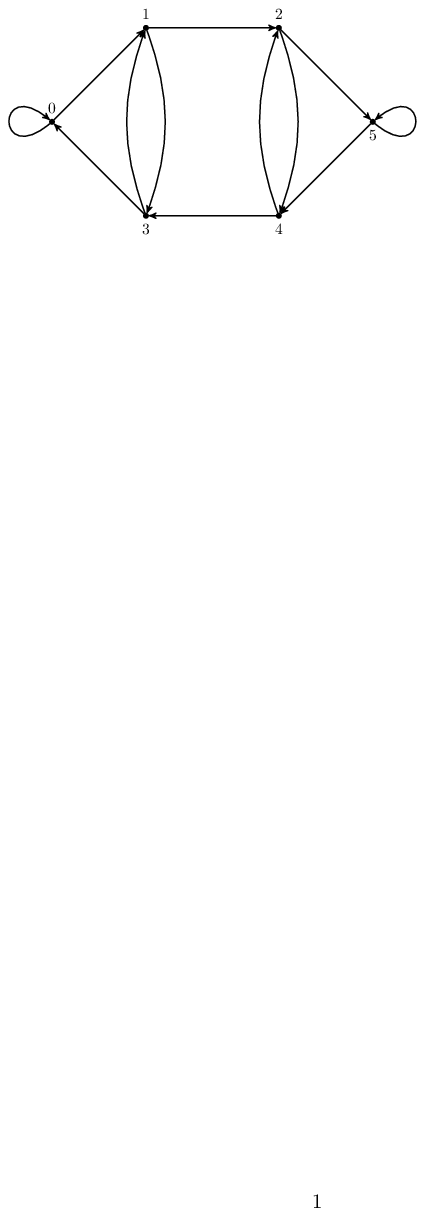}
	\caption{The binary graph with $V=3\cdot 2$ vertices and $B=3\cdot 2^2$ bonds.}
	\label{figBGp3}
\end{figure}

\begin{table}[htb]
	\caption[Comparison of the primitive pseudo orbit formula with numerics for a binary graph with $V=6$ vertices]{\label{figExTableVariance-2}
		For a binary graph with $6$ vertices the table shows the sizes of the sets of primitive pseudo orbits with no self-intersections and with one $2$-encounter of length zero along with the resulting variance for the first half of the coefficients of the characteristic polynomial. The last two columns provide numerical values of the variance for a set of uniformly distributed random bond lengths averaged over an interval of $50$ million mean spacings and the error between proposition \ref{thm:main} and the numerics.}
		\begin{indented}
		\item[]
	\begin{tabular}{ c c c c c c }
		$n$ & $|\mathcal{P}_0^n|$ & $|\widehat{\mathcal{P}}_{1}^n|$ 
		& $\langle |a_n|^2 \rangle_k$ 
		& Numerics & Error \\
		\hline
		0 & 1 & 0 & 1 & 1.000000 & 0.000000 \\
		1 & 2 & 0 & 1 & 1.000000 & 0.000000 \\
		2 & 3 & 0 & 3/4 & 0.750001 & -0.000001 \\
		3 & 6 & 0 & 3/4 & 0.750003 & -0.000003 \\
		4 & 10 & 4 & 7/8 & 0.874999 & 0.000001 \\
		5 & 8 & 4 & 1/2 & 0.499998 & 0.000002 \\
		6 & 8 & 8 & 3/8 & 0.374999 & 0.000001 \\
		\hline
	\end{tabular}
\end{indented}
\end{table}

\subsection{The Family of de Bruijn Graphs}
\label{sec:example family p=1}

We now consider the whole family of binary de Bruijn graphs with $p=1$ and  $V=2^r$; two graphs in this family are shown in figures \ref{figBGp1} and \ref{figBGp1-2}. 
The number of primitive pseudo orbits is $\PO_{1}(n) = \Lyndon$ for all $n$ and $C_{1} = 1$.
Therefore, by corollary \ref{cor:asymp}, in the semiclassical limit $r\to \infty$ we see $ \langle |a_n|^2 \rangle_{k} \approx 1/2$.
The numerically computed variance of the coefficients for $r = 3, 4, 5, 6$ is shown in figure \ref{fignumericsBGp1} and demonstrates the convergence of the variance to the expected asymptotic result.

\begin{figure}
	\centering
	\includegraphics[scale=1, trim=150 490 100 125, clip]{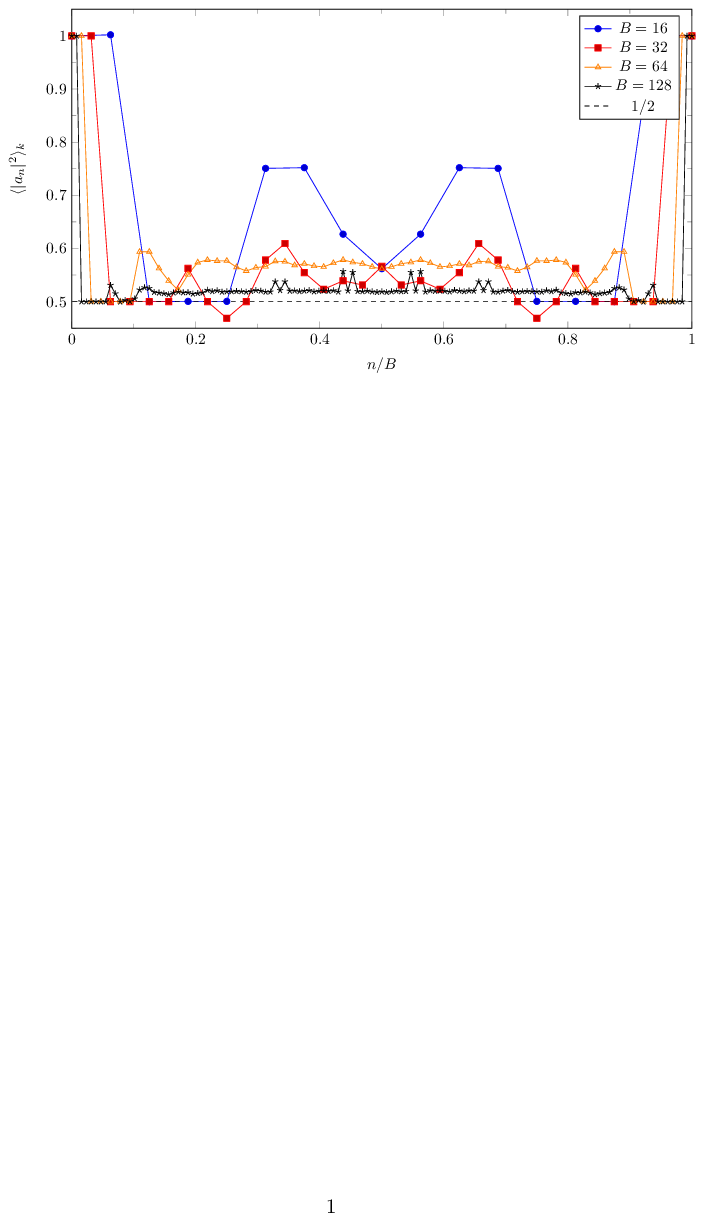}
	\caption{The numerically computed variance of the coefficients of the characteristic polynomial for the family of binary de Bruijn graphs, $p=1$. The plot shows the convergence of the variance to the asymptotic result $\langle |a_n|^2 \rangle_{k} \approx 1/2$.} 
	\label{fignumericsBGp1}
\end{figure}

\subsection{The Family of Binary Graphs with $p=3$}
\label{sec:example family p=3}

If we consider the family of binary graphs with $p=3$ and $V = 3\cdot 2^r$ vertices, the permutation $\pi$ associated to $\tilde{\A}_3$ has the expected fixed point associated to $\lambda = 2$ and a $2$-cycle associated to the primitive roots of unity $\pm 1$.  
Thus, $\PO_{3}(n) = \textrm{L}_2(n)$ when $n \neq 2$ and $\PO_{3}(2) = \textrm{L}_2(2) + 1$.
As a result, $C_{3} = 1 + 2^{-2} = 5/4$ and in the semiclassical limit 
$r\to \infty$ we have $ \langle |a_n|^2 \rangle_{k} \approx 5/8$ according to corollary \ref{cor:asymp}.
The numerically computed variance of the coefficients for $r = 1, 2, 3, 4, 5$ are shown in figure \ref{fignumericsBGp3}.  Again we see convergence to the asymptotic result consistent with corollary \ref{cor:asymp}.

\begin{figure}
	\centering
	\includegraphics[scale=1, trim=150 525 100 125, clip]{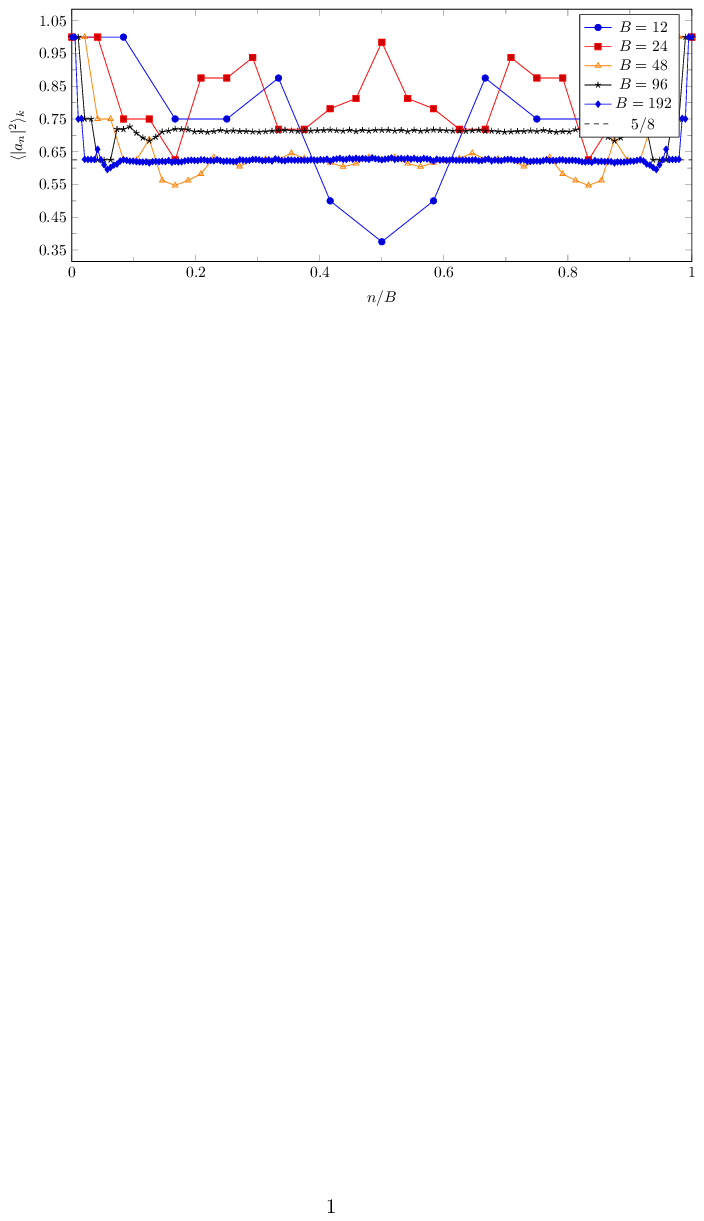}
	\caption{The numerically computed variance of the coefficients of the characteristic polynomial for the family of binary de Bruijn graphs, $p=3$. The plot shows the convergence of the variance to the asymptotic result $\langle |a_n|^2 \rangle_{k} \approx 5/8$.
	} 
	\label{fignumericsBGp3}
\end{figure}

\subsection{The Family of Binary Graphs with $p=5$}
\label{sec: example family p=5}

As a third example, we consider the family of binary graphs with $p=5$ and $V = 5\cdot 2^r$ vertices.
The binary graph from this family with 10 vertices is shown in figure \ref{figBGp5}.
The permutation $\pi$ associated to $\tilde{\A}_5$ has the fixed point associated with the eigenvalue $\lambda = 2$ and a $4$-cycle associated to the primitive fourth roots of unity.
Then $\PO_{5}(n) = \textrm{L}_2(n)$ when $n \neq 4$ and $\PO_{5}(4) = \textrm{L}_2(4) + 1$.
Thus, $C_{5} = 1 + 2^{-4} = 17/16$ and $\langle |a_n|^2 \rangle_{k} \approx 17/32$ for large  $r$.  
The numerically computed variance of the coefficients of the characteristic polynomial is shown in figure \ref{fignumericsBGp5} for $r = 2, 3, 4, 5$. Again we see convergence to the asymptotic result,  corollary \ref{cor:asymp}.

\begin{figure}[htb]
	\centering
	\includegraphics[scale=1, trim=210 500 50 155, clip]{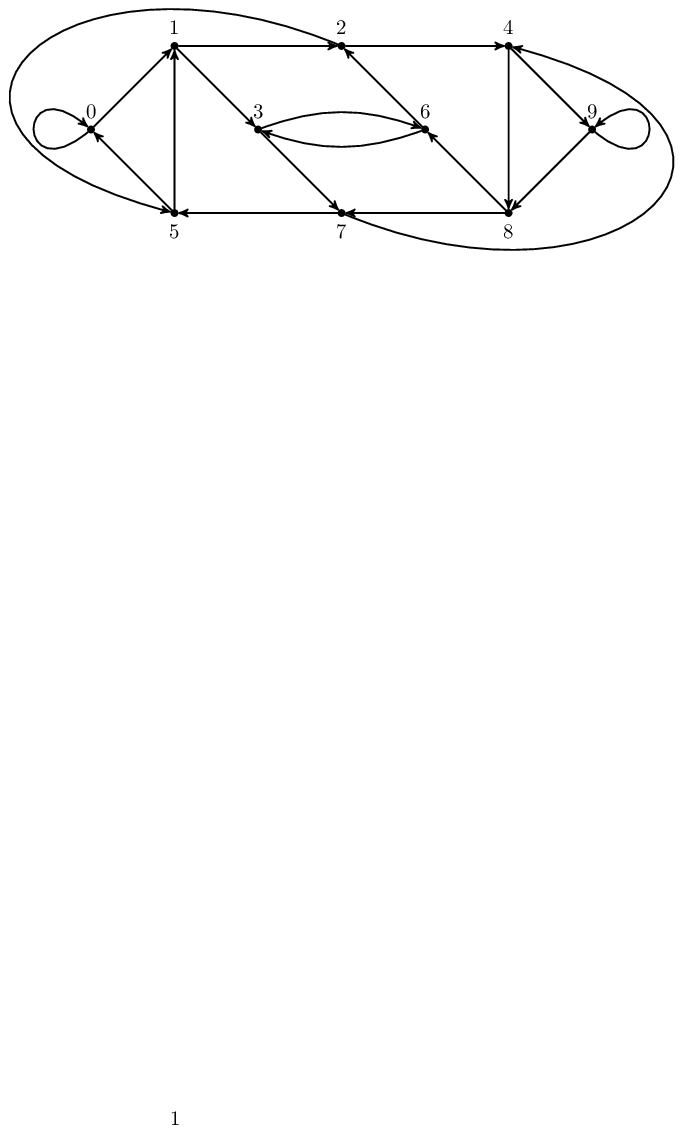}
	\caption{The binary graph with $V=5\cdot 2$ vertices and $B=5\cdot 2^2$ bonds} 
	\label{figBGp5}
\end{figure}

\begin{figure}
	\centering
	\includegraphics[scale=1, trim=150 525 100 125, clip]{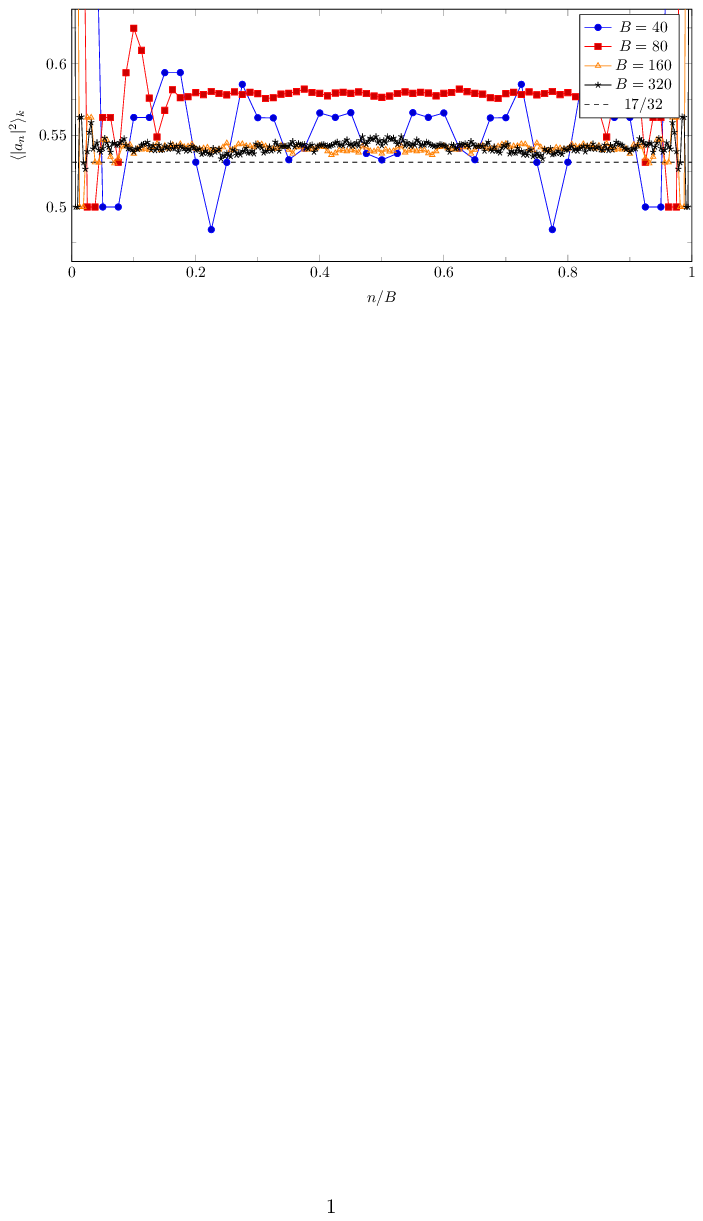}
	\caption{The numerically computed variance of the coefficients of the characteristic polynomial for the family of binary de Bruijn graphs with $p=5$. The plot shows the convergence of the variance to the asymptotic result $\langle |a_n|^2 \rangle_{k} \approx 17/32$.} 
	\label{fignumericsBGp5}
\end{figure}

\section{Conclusions}
\label{sec:conclusions}

We evaluated the first nontrivial moment, the variance, for coefficients of the characteristic polynomial of binary quantum graphs dynamically, proposition \ref{thm:main}.  The result depends on the number of primitive pseudo orbits with no self-intersections and the numbers of primitive pseudo orbits where all self-intersections are $2$-encounters of length zero.  This is the first example of a spectral statistic for a quantum system with chaotic classical dynamics where the periodic orbit formula can be evaluated without taking the semiclassical limit.  It is also important to note that the evaluation is exact; the contribution of all pseudo orbit pairs has been evaluated without using heuristic arguments to exclude any pseudo orbit pairs.

For the families of binary graphs with fixed $p$ and $V=p\cdot 2^r$ vertices and $B=p\cdot 2^{r+1}$ bonds, we evaluate the semiclassical limit of the variance, taking $r$ to infinity while fixing the ratio $n/B$, corollary \ref{cor:asymp}.  We find that the variance approaches a family-dependent constant independent of $n$.  The constant is determined from the total number of primitive pseudo orbits of length $n$, which in turn is evaluated from the cycle decomposition of a generalized $p\times p$ permutation matrix, proposition \ref{thmNumPPOs final}.  This explains how the exact variance formula agrees with approximate results obtained from diagonal arguments \cite{T02,BHS19}.
Such diagonal arguments are the starting point of most spectral analysis of quantum chaotic systems. However, we see that, for quantum binary graphs, many pseudo orbits which would contribute to a diagonal term do not appear in the formula for the variance. 
Instead the weighted sum of sets of primitive pseudo orbits where all self-intersections are of zero length turns out asymptotically to approach the total number of primitive pseudo orbits. 

In section \ref{sec:contributions} we introduced a cancellation mechanism to show that the contribution of most primitive pseudo orbits with self-intersections to the variance is zero. 
This is significant as the cancellation scheme does not follow the usual approach taken when using semiclassical arguments where the contributions of periodic orbits or pseudo orbits with a particular common structure are the classes that are grouped together to produce a contribution which can be evaluated. 
It would be interesting to see if such a cancellation mechanism extends to other quantum chaotic systems.

A first step to extend these results is to consider $q$-nary quantum graphs which are defined in a similar way to binary graphs with $V=p\cdot q^r$ vertices where $p$ and $q$ have no common factors.  The variance of $q$-nary graphs with $p=1$ was investigated in 
 \cite{BHS19} where a diagonal contribution was evaluated which agrees with numerical results in the semiclassical limit of large graphs.
  For binary graphs, the product of scattering amplitudes $A_{\bg} \, \overline{A}_{\bg'}$ is independent of how the scattering amplitudes are assigned at the vertices.  This is not the case for a $q$-nary graph.  However, it is possible to obtain some results by averaging over the assignment of vertex scattering amplitudes \cite{H20}.  It is worth noting that, while the argument used to count primitive pseudo orbits of length $n$ was explained for binary graphs, section \ref{sec:counting prim PPOs}, the same argument applies to $q$-nary graphs, see \cite{H20}.

One can also think of extending the scheme to other classes of quantum graphs.  While the binary graph structure provides a way to count orbits and pseudo orbits, proposition \ref{thm:main} holds for a regular directed graph with two incoming and two outgoing bonds at each vertex \cite{HH22}.  Obtaining the variance is then a matter of finding classes of $4$-regular graphs where it is possible to count primitive periodic orbits and primitive pseudo orbits.

\ack
The authors would like to thank Gregory Berkolaiko for helpful comments.
	JH would also like to thank Mark Pollicott and the University of Warwick for their
	hospitality during his sabbatical where some of the work was carried out, and where JH was
	supported by the Baylor University research leave program. 
	TH would like to thank Baylor University and the University of Dallas where some of the work was carried out.
	This work was partially
	supported by a grant from the Simons Foundation (354583 to Jonathan Harrison).

\appendix

\section{Pseudo Orbits for the Binary de Bruijn graph with $V=8$.}
\label{sec:appendixA}

For the binary de Bruijn graph with $V=8$ vertices, shown in figure \ref{figBGp1-2}, the variance of the coefficients of the graph's characteristic polynomial were evaluated in table \ref{figExTableVariance}.
Recall that for a binary de Bruijn graph, the primitive periodic orbits of length $n$ correspond uniquely to binary Lyndon words of length $n$.
Moreover, the primitive pseudo orbits of length $n$ correspond uniquely to the binary words of length $n$ that have a strictly decreasing Lyndon decomposition.
Table \ref{figExTableOrbits} lists the sizes of all sets of primitive pseudo orbits of the lengths corresponding to the first half of the characteristic polynomial's coefficients.
Note that, $\mathcal{P}_0^0 = \{ \bar{0} \}$, the empty pseudo orbit.

\begin{table}[htb!]
	\caption{Sizes of the sets of primitive pseudo orbits that appear in proposition \ref{thm:main} for a binary graph with $V=8$ vertices.}
	\centering
	\begin{tabular}{ c c c c c c }
		$n$ & $\PPO_2(n)$ & $|\mathcal{P}_0^n|$ & $|\widehat{\mathcal{P}}_{1}^n|$ 
		& $|\widehat{\mathcal{P}}_{2}^n|$ & $| \{ \bg : B_{\bg} = n, C_{\bg} = 0 \} |$ \\
		\hline
		0 & 1 & 1 & 0 & 0 & 0 \\
		1 & 2 & 2 & 0 & 0 & 0 \\
		2 & 2 & 2 & 0 & 0 & 0 \\
		3 & 4 & 4 & 0 & 0 & 0 \\
		4 & 8 & 8 & 0 & 0 & 0 \\
		5 & 16 & 8 & 8 & 0 & 0 \\
		6 & 32 & 8 & 20 & 0 & 4 \\
		7 & 64 & 16 & 16 & 8 & 24 \\
		8 & 128 & 16 & 16 & 24 & 72 \\
		\hline
	\end{tabular}
	\label{figExTableOrbits}
\end{table}

The non-empty sets of strictly decreasing Lyndon decompositions corresponding to primitive pseudo orbits used to produce table \ref{figExTableOrbits} are listed below. 
\begin{eqnarray}
\mathcal{P}_0^1 &= \{ (0), \ (1) \} \\
\mathcal{P}_0^2 &= \{ (01), \  (1)(0) \} \\
\mathcal{P}_0^3 &= \{ (001), \ (01)(0), \ (011), \  (1)(01) \} \\
\mathcal{P}_0^4 &= \{ (0001), \ (001)(0), \ (0011), \ (011)(0), \nonumber \\
&\qquad (0111), \ (1)(001), \ (1)(01)(0), \ (1)(011) \} \\
\mathcal{P}_0^5 &= \{ (00011), \ (0011)(0), \ (00111), \ (0111)(0), \nonumber \\
&\qquad (1)(0001), \ (1)(001)(0), \ (1)(0011), \ (1)(011)(0) \} \\
\widehat{\mathcal{P}}_{1}^5 &= \{ (00001), \ (0001)(0), \ (00101), \ (01)(001), \nonumber \\
&\qquad (01011), \  (011)(01), \ (01111), \ (1)(0111) \} \\
\mathcal{P}_0^6 &= \{ (000111), \ (001011), \ (001101), \ (1)(00011), \nonumber \\
&\qquad (00111)(0), \ (01)(0011), \ (1)(0011)(0), \ (011)(001) \} \\
\widehat{\mathcal{P}}_{1}^6 &= \{ (000011), \ (000101), \ (001111), \ (010111), \nonumber \\
&\qquad (1)(00001), \ (00011)(0), \ (00101)(0), \ (1)(00101), \nonumber \\
&\qquad (1)(00111), \ (01011)(0), \ (1)(01011), \ (01111)(0),  \nonumber \\
&\qquad (01)(0001), \ (0111)(01), \ (1)(0001)(0), \ (1)(0111)(0), \nonumber \\
&\qquad (01)(001)(0), \ (1)(01)(001), \ (011)(01)(0), \ (1)(011)(01) \} 
\end{eqnarray}
\begin{equation}
\fl \{ \bg : B_{\bg} = 6, C_{\bg} = 0 \} = \{
(000001), \ (00001)(0), \ (011111), \ (1)(01111) \}
\end{equation}
\begin{eqnarray}
\mathcal{P}_0^7 &= \{ (0001011), \ (0001101), \ (0 0 1 0 1 1)( 0), \ (0 0 1 0 1 1 1 ), \nonumber \\
&\qquad (0 0 1 1 0 1 )(0), \ (0 0 1 1 1 0 1 ) , \ (0 1)( 0 0 0 1 1), \ (0 1)( 0 0 1 1)( 0), \nonumber \\
&\qquad  (0 1)( 0 0 1 1 1 ), \ (0 1 1)( 0 0 0 1), \ (0 1 1)( 0 0 1)( 0), \ (0 1 1 1 )(0 0 1 ), \nonumber \\
&\qquad (1) (0 0 1 0 1 1 ), \  (1)( 0 0 1 1 0 1), \ (1)( 0 1 )(0 0 1 1), \ (1)( 0 1 1 )(0 0 1)  \} \\
\widehat{\mathcal{P}}_{1}^7 &= \{ (0000111), \ (000111)(0), \ (0001111), \ (0 0 1 1 1 1 )(0), \nonumber \\
&\qquad (0 1 0 1 1 1)( 0), \ (0 1 1 1)( 0 1 )(0 ), \ (1)( 0 0 0 0 1 1), \ (1)( 0 0 0 1 0 1) , \nonumber \\
&\qquad (1 )(0 0 0 1 1)( 0), \  (1)( 0 0 0 1 1 1 ), \ (1)( 0 0 1 0 1)( 0), \ (1)( 0 0 1 1 1)( 0), \  \nonumber \\
&\qquad(1 )(0 1)( 0 0 0 1), \ (1 )(0 1)( 0 0 1)( 0), \ (1)( 0 1 0 1 1)( 0) , \nonumber \\ 
&\qquad (1)( 0 1 1)( 0 1 )(0)  \} \\
\widehat{\mathcal{P}}_{2}^7 &= \{ (0000101), (000101)(0), \ (0 1)( 0 0 0 0 1), \ (0 1) (0 0 0 1)( 0), \nonumber \\
&\qquad (0 1 0 1 1 1 1 ), \ (0 1 1 1 1)( 0 1 ), \ (1)( 0 1 0 1 1 1 ), \ (1)( 0 1 1 1)( 0 1) \} 
\end{eqnarray}
%
\begin{eqnarray}
\fl \{ \bg : B_{\bg} = 7, C_{\bg} = 0 \} & =  \{ (0000001), \ (000001)(0), \ (0000011), \ (000011)(0),  \nonumber \\
&\qquad(0001001), \ (0 0 1) (0 0 0 1), \ (0 0 1)( 0 0 1 1), \ (0 0 1 0 1 0 1), \nonumber \\
&\qquad(0 0 1 1 )(0 0 1 ), \ (0 0 1 1 0 1 1 ), \  (0 0 1 1 1 1 1 ), \ (0 1)( 0 0 1 0 1 ),  \nonumber \\
&\qquad(0 1 0 1 0 1 1 ), \ (0 1 0 1 1)( 0 1 ), \ (0 1 1)( 0 0 1 1), \ (0 1 1 0 1 1 1), \nonumber \\
&\qquad(0 1 1 1)( 0 1 1 ), \ (0 1 1 1 1 1 )(0), \  (0 1 1 1 1 1 1), \ (1)( 0 0 0 0 0 1),  \nonumber \\
&\qquad(1)( 0 0 0 0 1 )(0), \ (1 )(0 0 1 1 1 1), \  (1)( 0 1 1 1 1)( 0 ), \nonumber \\
&\qquad(1)( 0 1 1 1 1 1) \} 
\end{eqnarray}
%
%
\begin{eqnarray}
\mathcal{P}_0^8 &= \{ 
(0 0 0 1 0 1 1 1), \
(0 0 0 1 1 1 0 1), \ 
(0 0 1 0 1 1 1)( 0), \ 
(0 0 1 1 1 0 1 )(0), \nonumber \\
&\qquad(0 1)( 0 0 0 1 1 1), \
(0 1)( 0 0 1 1 1 )(0), \
(0 1 1 1)( 0 0 0 1), \ 
(0 1 1 1)( 0 0 1)( 0), \nonumber \\
&\qquad(1)( 0 0 0 1 0 1 1), \ 
(1)( 0 0 0 1 1 0 1), \ 
(1)( 0 0 1 0 1 1)( 0), \ 
(1)( 0 0 1 1 0 1)( 0), \nonumber  \\
&\qquad(1)( 0 1)( 0 0 0 1 1), \ 
(1)( 0 1)( 0 0 1 1)( 0), \
(1)( 0 1 1)( 0 0 0 1 ), \nonumber \\
&\qquad(1)( 0 1 1)( 0 0 1)( 0) \}
\\
\widehat{\mathcal{P}}_{1}^8 &= \{ 
(0 0 0 0 1 0 1 1 ), \  
(0 0 0 0 1 1 0 1), \     
(0 0 0 1 0 1 1)( 0), \  
(0 0 0 1 1 0 1)( 0), \nonumber \\
&\qquad(0 0 1 0 1 1 1 1), \     
(0 0 1 1 1 1 0 1 ), \     
(0 1)( 0 0 0 0 1 1), \
(0 1)( 0 0 0 1 1)( 0), \nonumber \\
&\qquad(0 1)( 0 0 1 1 1 1 ), \  
(0 1 1) (0 0 0 0 1), \ 
(0 1 1 )(0 0 0 1 )(0 ), \ 
(0 1 1 1 1)( 0 0 1), \nonumber \\ 
&\qquad(1)( 0 0 1 0 1 1 1), \ 
(1)( 0 0 1 1 1 0 1), \ 
(1)( 0 1)( 0 0 1 1 1), \nonumber \\ 
&\qquad(1)( 0 1 1 1)( 0 0 1) \} \ 
\\
\widehat{\mathcal{P}}_{2}^8 &= \{ 
(0 0 0 0 1 1 1 1 ), \ 
(0 0 0 1 0 0 1 1 ), \ 
(0 0 0 1 1 0 0 1 ), \ 
(0 0 0 1 1 1 1)( 0), \nonumber \\
&\qquad(0 0 1 )(0 0 0 1 1), \ 
(0 0 1 0 1 1 0 1), \ 
(0 0 1 1 )(0 0 0 1), \ 
(0 0 1 1 0 1 1 1), \nonumber \\
&\qquad(0 0 1 1 1 0 1 1 ), \ 
(0 1 0 1 1 )(0 0 1), \ 
(0 1 0 1 1 1 1)( 0), \ 
(0 1 1)( 0 0 1 0 1), \nonumber \\
&\qquad(0 1 1)( 0 0 1 1 1), \ 
(0 1 1)( 0 1)( 0 0 1), \ 
(0 1 1 1)( 0 0 1 1), \ 
(0 1 1 1 1)( 0 1)( 0), \nonumber \\ 
&\qquad(1)( 0 0 0 0 1 0 1 ), \ 
(1)( 0 0 0 0 1 1 1), \
(1)( 0 0 0 1 0 1)( 0), \ 
(1)( 0 0 0 1 1 1)( 0), \nonumber \\ 
&\qquad(1)( 0 1 )(0 0 0 0 1 ), \ 
(1)( 0 1)( 0 0 0 1)( 0), \ 
(1)( 0 1 0 1 1 1)( 0), \nonumber \\ 
&\qquad(1)( 0 1 1 1)( 0 1)( 0) \}
\end{eqnarray}
%
\begin{eqnarray}
\fl \{ \bg : B_{\bg} = 8, C_{\bg} = 0 \} &=  \{ 
(0 0 0 0 0 0 0 1 ), \ 
(0 0 0 0 0 0 1)( 0), \ 
(0 0 0 0 0 0 1 1 ), \ 
(0 0 0 0 0 1 0 1 ), \nonumber \\ 
&\qquad(0 0 0 0 0 1 1)( 0), \ 
(0 0 0 0 0 1 1 1), \  
(0 0 0 0 1 0 0 1), \
(0 0 0 0 1 0 1)( 0), \nonumber \\ 
&\qquad(0 0 0 0 1 1 1 )(0), \ 
(0 0 0 1 0 0 1)( 0), \
(0 0 0 1 0 1 0 1), \  
(0 0 0 1 1 0 1 1 ), \nonumber \\ 
&\qquad(0 0 0 1 1 1 1 1), \ 
(0 0 1)( 0 0 0 0 1), \ 
(0 0 1)( 0 0 0 1 )(0), \ 
(0 0 1 0 0 1 0 1), \nonumber \\ 
&\qquad(0 0 1 0 0 1 1)( 0), \ 
(0 0 1 0 0 1 1 1), \ 
(0 0 1 0 1)( 0 0 1), \ 
(0 0 1 0 1 0 1 )(0 ), \nonumber \\ 
&\qquad(0 0 1 0 1 0 1 1), \ 
(0 0 1 1 )(0 0 1)( 0), \ 
(0 0 1 1 0 1 0 1), \ 
(0 0 1 1 0 1 1 )(0), \nonumber \\ 
&\qquad(0 0 1 1 1)( 0 0 1), \  
(0 0 1 1 1 1 1)( 0), \ 
(0 0 1 1 1 1 1 1), \ 
(0 1)( 0 0 0 0 0 1), \nonumber \\ 
&\qquad(0 1)( 0 0 0 0 1 )(0 ), \ 
(0 1)( 0 0 0 1 0 1), \ 
(0 1)( 0 0 1 0 1)( 0), \nonumber \\ 
&\qquad(0 1)( 0 0 1 0 1 1), \
(0 1)( 0 0 1 1 0 1 ), \ 
(0 1 0 1 0 1 1)( 0), \ 
(0 1 0 1 0 1 1 1 ), \nonumber \\
&\qquad(0 1 0 1 1)( 0 1)( 0), \ 
(0 1 0 1 1 0 1 1 ), \  
(0 1 0 1 1 1)( 0 1), \ 
(0 1 0 1 1 1 1 1), \nonumber \\ 
&\qquad(0 1 1)( 0 0 0 1 1), \ 
(0 1 1 )(0 0 1 1)( 0), \ 
(0 1 1)( 0 1 0 1 1), \ 
(0 1 1 0 1 1 1)( 0), \nonumber \\  
&\qquad(0 1 1 0 1 1 1 1 ), \ 
(0 1 1 1)( 0 1 1)( 0), \ 
(0 1 1 1 1)( 0 1 1), \ 
(0 1 1 1 1 1)( 0 1), \nonumber \\ 
&\qquad(0 1 1 1 1 1 1)( 0), \ 
(0 1 1 1 1 1 1 1), \  
(1) (0 0 0 0 0 0 1 ), \ 
(1)( 0 0 0 0 0 1 )(0), \nonumber \\ 
&\qquad(1)( 0 0 0 0 0 1 1), \
(1)( 0 0 0 0 1 1)( 0), \ 
(1)( 0 0 0 1 0 0 1), \ 
(1)( 0 0 0 1 1 1 1), \nonumber \\   
&\qquad(1)( 0 0 1 )(0 0 0 1), \ 
(1)( 0 0 1 0 0 1 1), \ 
(1)( 0 0 1 0 1 0 1), \ 
(1)( 0 0 1 1 )(0 0 1), \nonumber \\ 
&\qquad(1)( 0 0 1 1 0 1 1), \ 
(1)( 0 0 1 1 1 1)( 0), \  
(1)( 0 0 1 1 1 1 1), \ 
(1)( 0 1)( 0 0 1 0 1), \nonumber \\ 
&\qquad(1)( 0 1 0 1 0 1 1), \ 
(1)( 0 1 0 1 1 )(0 1), \ 
(1)( 0 1 0 1 1 1 1), \
(1)( 0 1 1)( 0 0 1 1), \nonumber \\ 
&\qquad(1)( 0 1 1 0 1 1 1), \
(1)( 0 1 1 1 )(0 1 1, \
(1)( 0 1 1 1 1)( 0 1), \nonumber \\  
&\qquad(1)( 0 1 1 1 1 1)( 0), \
(1)( 0 1 1 1 1 1 1) \} 
\end{eqnarray}

\section{Pseudo Orbits for the Binary Graph with $V=6$.}
\label{sec:appendixB}

For the binary graph with $V=6$ vertices shown in figure \ref{figBGp3}, the variance of the coefficients of the graph's characteristic polynomial were evaluated in table \ref{figExTableVariance-2}.
For the binary graph family with $p=3$, primitive periodic orbits of length $n$ are not in bijection with binary Lyndon words of length $n$.
Consequently, we use words over the vertex label alphabet $\mathcal{V} = \{ 0, 1, 2, 3, 4, 5 \}$ to represent a closed path; rotations of the word are closed paths that belong to the same periodic orbit.
A word of length $n$ corresponds to a closed path of length $n$.
For example, the word $013$ labels the path of length of three that traverses vertices 0, 1, 3, and returns to 0.
We use parentheses to mark different primitive periodic orbits in a primitive pseudo orbit.
Table \ref{figExTableOrbits-2} shows the sizes of the sets of primitive pseudo orbits that appear in proposition \ref{thm:main} with lengths corresponding to the first half of the characteristic polynomial's coefficients.
As $C_3= 5/4$, the number of primitive pseudo orbits of length $n$ is $\PPO_3(n) = 5\cdot 2^{n-3}$ for $n > 3$, by proposition \ref{thmNumPPOs}.
Note that, $\mathcal{P}_0^0 = \{ \bar{0} \}$, the empty pseudo orbit.

\begin{table}[htb!]
	\caption{Sizes of the sets of primitive pseudo orbits that appear in proposition \ref{thm:main} for a binary graph with $V=6$ vertices.}
	\centering
	\begin{tabular}{ c c c c c c }
		$n$ & $\PPO_2(n)$ & $|\mathcal{P}_0^n|$ & $|\widehat{\mathcal{P}}_{1}^n|$ 
		& $| \{ \bg : B_{\bg} = n, C_{\bg} = 0 \} |$ \\
		\hline
		0 & 1 & 1 & 0 & 0 \\
		1 & 2 & 2 & 0 & 0 \\
		2 & 3 & 3 & 0 & 0 \\
		3 & 6 & 6 & 0 & 0 \\
		4 & 10 & 6 & 4 & 0 \\
		5 & 20 & 8 & 4 & 8 \\
		6 & 40 & 8 & 8 & 24 \\
		\hline
	\end{tabular}
	\label{figExTableOrbits-2}
\end{table}

The non-empty sets of primitive pseudo orbits in table \ref{figExTableOrbits-2}, are listed below. 
\begin{eqnarray}
\mathcal{P}_0^1 &= \{ (0), \ (5) \} \\
\mathcal{P}_0^2 &= \{ (13), \  (24), \ (5)(0) \} \\
\mathcal{P}_0^3 &= \{ (013), \ (13)(0), \ (24)(0), \ (254),  \  (5)(13),  \ (5)(24) \} \\
\mathcal{P}_0^4 &= \{ (1243), \ (24)(13), \ (254)(0), \ (5)(013), \ (5)(13)(0), \ (5)(24)(0) \} 
\end{eqnarray}
\begin{equation}
\widehat{\mathcal{P}}_0^4 = \{
(0013), \ (013)(0), \ (2554), \ (5)(254) \}
\end{equation}
\begin{eqnarray} 
\mathcal{P}_0^5 &= \{ 
(01243), \ (12543), \ (1243)(0), \ (5)(1243), \nonumber \\
&\qquad(24)(013), \ (254)(13), \ (24)(13)(0), \ (5)(24)(13) \} \\
\widehat{\mathcal{P}}_{1}^5 &= \{ (5)(0013), \ (2554)(0), \ (5)(013)(0), \ (5)(254)(0) \} 
\end{eqnarray}
\begin{eqnarray}
\fl \{ \bg : B_{\bg} = 5, C_{\bg} = 0 \} = \{
&(00013), \ (0013)(0), \ (01313), \ (13)(013), \nonumber \\
&(24254), \ (254)(24), \ (25554), \ (5)(2554) \}
\end{eqnarray}
\begin{eqnarray}
\mathcal{P}_0^6 &= \{ (012543), \ (12543)(0), \ (254)(013), \ (254)(13)(0), \nonumber \\
&\qquad(5)(01243), \ (5)(1243)(0), \ (5)(24)(013), \ (5)(24)(13)(0) \} \\
\widehat{\mathcal{P}}_{1}^6 &= \{ 
(001243), \ (01243)(0), \ (24)(013)(0), \ (24)(0013), \nonumber \\ 
&\qquad (2554)(13), \ (5)(12543), \ (5)(254)(24), \ (554312),  \} 
\end{eqnarray}
\begin{eqnarray}
\fl \{ \bg : B_{\bg} = 6, C_{\bg} = 0 \} &= \{
(000013), \ (00013)(0), \ (001313), \ (01313)(0), \nonumber\\
&\qquad(124243), \ (124313), \ (13)(0013), \ (13)(013)(0), \nonumber \\
&\qquad(13)(1243), \ (24)(1243), \ (24254)(0), \ (254)(24)(0), \nonumber \\
&\qquad(2554)(24), \ (25554)(0), \ (5)(00013), \ (5)(0013)(0), \nonumber \\
&\qquad(5)(01313), \ (5)(13)(013), \ (5)(24254), \ (5)(254)(24), \nonumber \\
&\qquad(5)(2554)(0), \ (5)(25554), \ (554242), \ (555542) \}
\end{eqnarray}



\section*{References}

\begin{thebibliography}{99}
\bibitem{Aetal00}
Akkermans E, Comtet A, Debois J, Montanbaux G and Texier C
2000
Spectral determinant on quantum graphs
{\it Ann. Phys.}
{\bf 284}
10--51

\bibitem{BHJ12}
Band R, Harrison JM and CH Joyner
2012
Finite pseudo orbit expansions for spectral quantities of quantum graphs
\jpa
{\bf 45}
325204

\bibitem{BHS19}
Band R, Harrison JM and Sepanski M
2019
Lyndon word decompositions and pseudo orbits on $q$-nary graphs
{\it J. Math. Anal. Appl.}
{\bf 470}
135--144

\bibitem{BBK01}
Berkolaiko G, Bogomolny EB and Keating JP 
2001
Star graphs and \v{S}eba billiards 
\JPA 
{\bf 34}
335--350

\bibitem{BHN08}
Berkolaiko G, Harrison JM and Novaes M
2008
Full counting statistics of chaotic cavities from classical action correlations
\jpa
{\bf 41}
365102

\bibitem{BK13}
Berkolaiko G and Kuchment P
2013
{\it Introduction to Quantum Graphs (Mathematical Surveys and Monographs}
vol~186)
(Providence, RI: American Mathematical Society)

\bibitem{BK12}
Berkolaiko G and Kuipers J 
2012
Universality in chaotic quantum transport: The concordance between random matrix and semiclassical theories 
{\it Phys. Rev. E}
{\bf 85} 
045201(R)

\bibitem{BSW02} 
Berkolaiko G, Schanz H and Whitney RS
2002
Leading off-diagonal correction to the form factor of large graphs
\PRL 
{\bf 88}
104101

\bibitem{BSW03}
Berkolaiko G, Schanz H and Whitney RS
2003
Form factor for a family of quantum graphs: an expansion to third order
\JPA
{\bf 36}
8373--8392

\bibitem{B85}
Berry MV
1985
Semiclassical theory of spectral rigidity
{\it Proc. R. Soc. A}
{\bf 400}
229–-251

\bibitem{BK90}
Berry MV and Keating JP
1990 
A rule for quantizing chaos? 
{\JPA}
{\bf 23}
4839--4849

\bibitem{BK92}
Berry MV and Keating JP 
1992
A new asymptotic representation for $\zeta (1/2+\rmi t)$ and quantum spectral determinants 
{\it Proc. R. Soc. Lond. A} 
{\bf 437}
151--173



\bibitem{B92}
Bogomolny E
1992
Semiclassical quantization of multidimensional systems 
{\it Nonlinearity}
{\bf 5} 
805


\bibitem{BE09}
Bolte J and Endres S
2009
The trace formula for quantum graphs with general self adjoint boundary conditions
{\it Ann. Henri Poincar{\'{e}}}
{\bf 10}
189--223 

\bibitem{BH03}
Bolte J and Harrison JM
2003
The spin contribution to the form factor of quantum graphs
\JPA
{\bf 36}
L433--L440

\bibitem{CFL58}
Chen KT, Fox RH and Lyndon RC
1958
Free differential calculus, IV
{\it Ann. Math.}
{\bf 68}
81--95

\bibitem{D79}
Davis PJ, 
1979 
{\it Circulant Matrices} 
(New York, NY: Wiley)

\bibitem{DOW05}
Degli Esposti M, O'Keefe S, Winn B 
2005 
A semi-classical study of the Casati-Prosen triangle map
{\it Nonlinearity} 
{\bf 18}
1073--1094

\bibitem{D83}
Duval JP
1983
Factorizing Words over an Ordered Alphabet
{\it J. Algorithms}
{\bf 4} 
363--381

\bibitem{F19-1}
Faal HT
2019
A multiset version of determinants and the coin arrangements lemma
{\it Theor. Comput. Sci.}
{\bf 793}
36--43	

\bibitem{F19-2}
Faal HT
2019
A multiset version of even-odd permutations identity
{\it Int. J. Found. Comput.}
{\bf 30}
683--691

\bibitem{GM15}
Garcia-Planas MI and Magret M
2015
Eigenvalues and eigenvectors of monomial matrices
{\it Proceedings of the XXIV Congress on Differential Equations and Applications} 
963--966

\bibitem{GS06}
Gnutzmann S and Smilansky U
2006
Quantum graphs: applications to quantum chaos and universal spectral statistics
{\it Adv. Phys.}
{\bf 55}
527--625

\bibitem{G71}
Gutzwiller MC
1971
Periodic orbits and classical quantization conditions 
{\it J. Math. Phys.} 
{\bf 12} 
343-–358

\bibitem{G90}    
Gutzwiller MC
1990
{\it Chaos in Classical and Quantum Mechanics} 
(Springer, New York)

\bibitem{HA84}
Hannay JH and Ozorio de Almeida AM 
1984
Periodic orbits and a correlation function for the semiclassical density of states 
\JPA
{\bf 17}
3429--3440


\bibitem{HSW07}
Harrison JM, Smilansky U and Winn B
2007
Quantum graphs where back-scattering is prohibited
\jpa
{\bf 40}
14181--14193

\bibitem{HH22}
Harrison JM and Hudgins TK
2022
Periodic-orbit evaluation of a spectral statistic of quantum graphs without the semiclassical limit
{\it EPL}
{\bf 138}
30002

\bibitem{H20}
Hudgins TK
2020
Orbits, pseudo orbits, and the characteristic polynomial of $q$-nary quantum graphs
{\it ProQuest Dissertations Publishing}
28028730

\bibitem{K92}
Keating JP 
1992
Periodic orbit resummation and the quantumzation of chaos 
{\it Proc. R. Soc. A} 
{\bf 436}
99--108

\bibitem{KM00}
Keating JP and Mezzadri F
2000
Pseudo-symmetries of Anosov maps and spectral statistics
{\it Nonlinearity}
{\bf 13}
747--775

\bibitem{KMM01}
Keppeler S, Marklof J and Mezzadri F
2001
Quantum cat maps with spin $1/2$
{\it Nonlinearity}
{\bf 14}
719--738

\bibitem{KPS07}
Kostrykin V, Potthoff J, and Schrader R 
2007
Heat kernels on metric graphs and a trace formula
Adventure in Mathematical Physics (F. Germinet and P.D. Hislop, eds.)
{\it Contemp. Math.}
{\bf 447} (Amer. Math. Soc., Providence, RI)
175--198

\bibitem{KS99}
Kottos T and Smilansky U
1999
Periodic orbit theory and spectral statistics for quantum graphs
{\it Ann. Phys.} 
{\bf 274}
76--124

\bibitem{KOR14}
Kurasov P, Ogik R and Rauf A
2014
On reflectionless equi-transmitting matrices
{\it Opuscula Math.}
{\bf 34}
483--501

\bibitem{L83}
Lothaire M 
1983
{\it Combinatorics on Words}
(Reading, Massachusetts: Addison-Wesley)

\bibitem{Metal04}
M\"uller S, Heusler S, Braun P, Haake F, and Altland A
2004
Semiclassical foundation of universality in quantum chaos
\PRL
{\bf 93}
014103

\bibitem{Metal05}
M\"uller S, Heusler S, Braun P, Haake F, and Altland A
2005
Periodic-orbit theory of universality in quantum chaos
{\it Phys. Rev. E}	
{\bf 72}
046207

\bibitem{NM09}
Nagao T and M\"uller
2009
The $n$-level spectral correlations for chaotic systems 	
{\jpa}	
{\bf 42}
375102


\bibitem{R83}
Roth J-P 
1983
Spectre du Laplacien sur un graphe 
{\it C. R. Acad. Sci. Paris S\'er. I Math.} 
{\bf 296}
793--795


\bibitem{S60}
Sherman S
1960
Combinatorial aspects of the Ising model for ferromagnetism.  I. A conjecture of Feynman on paths and graphs
\JMP
{\bf 1}
202--217

\bibitem{S02} 
Sieber M 
2002 
Leading off-diagonal approximation for the spectral form factor for uniformly hyperbolic systems
\JPA
{\bf 35} 
L613--L619

\bibitem{SR01}
Sieber M and Richter K
2001	
Correlations between periodic orbits and their role in spectral statistics
\PS
{\bf T90}
128

\bibitem{T00}
Tanner G
2000	
Spectral statistics for unitary transfer matrices of binary graphs
\JPA
{\bf 33}
3567--3585

\bibitem{T01}
Tanner G
2001
Unitary-stochastic matrix ensembles and spectral statistics
\JPA
{\bf 34}
8485--8500

\bibitem{T02}
Tanner G
2002	
The autocorrelation function for spectral determinants of quantum graphs
\JPA
{\bf 35}
5985--5995

\bibitem{TC11}
Turek O and Cheon T
2011
Quantum graph vertices with permutation-symmetric scattering probabilities
{\it Phys. Lett. A}
{\bf 375}
3775--3780	

\end{thebibliography}
\end{document}